\documentclass[12pt]{amsart}
\usepackage{amsmath}
\usepackage{amsfonts}
\usepackage{amssymb}
\usepackage[all,cmtip]{xy}           
\usepackage{bm}
\usepackage{bbm}
\usepackage{bbding}
\usepackage{txfonts}
\usepackage{amscd}
\usepackage{xspace}
\usepackage[shortlabels]{enumitem}
\usepackage{ifpdf}

\ifpdf
  \usepackage[colorlinks,final,backref=page,hyperindex]{hyperref}
\else
  \usepackage[colorlinks,final,backref=page,hyperindex,hypertex]{hyperref}
\fi
\usepackage{tikz}
\usepackage[active]{srcltx}

\usepackage{tikz-cd}
\usepackage{tikz}

\makeatletter

\newtheorem{thm}{Theorem}[section]

\newtheorem{cor}[thm]{Corollary}
\newtheorem{pro}[thm]{Proposition}
\theoremstyle{definition}
\newtheorem{ex}[thm]{Example}
\newtheorem{rmk}[thm]{Remark}
\newtheorem{defi}[thm]{Definition}

\newcommand{\nc}{\newcommand}
\newcommand{\delete}[1]{}

\nc{\mlabel}[1]{\label{#1}}  
\nc{\mcite}[1]{\cite{#1}}  
\nc{\mref}[1]{\ref{#1}}  
\nc{\meqref}[1]{\eqref{#1}}  
\nc{\mbibitem}[1]{\bibitem{#1}} 

\delete{
\nc{\mlabel}[1]{\label{#1}{\hfill \hspace{1cm}{\bf{{\ }\hfill(#1)}}}}
\nc{\mcite}[1]{\cite{#1}{{\bf{{\ }(#1)}}}}  
\nc{\mref}[1]{\ref{#1}{{\bf{{\ }(#1)}}}}  
\nc{\meqref}[1]{\eqref{#1}{{\bf{{\ }(#1)}}}}  
\nc{\mbibitem}[1]{\bibitem[\bf #1]{#1}} 
}

\DeclareMathOperator{\im}{Im}

\newcommand {\comment}[1]{{\marginpar{*}\scriptsize\textbf{Comments:} #1}}
\newcommand {\emptycomment}[1]{}
\newcommand {\yh}[1]{{\textcolor{purple}{yh: #1}}}

\newcommand{\hl}[1]{\textcolor{blue}{hl: #1}}

\delete{
\newcommand {\yh}[1]{{yh: #1}}

\newcommand{\hl}[1]{{hl: #1}}

}

\nc{\oprn}{\theta}

\delete{
\setlength{\baselineskip}{1.8\baselineskip}
\newcommand{\comment}[1]{{\marginpar{*}\scriptsize\textbf{Comments:} #1}}
\newcommand{\emptycomment}[1]{}
\newcommand{\yh}[1]{{\marginpar{*}\scriptsize\textcolor{purple}{yh: #1}}}

}

\nc{\calo}{\mathcal{O}}
\nc{\oop}{$\mathcal{O}$-operator\xspace}
\nc{\oops}{$\mathcal{O}$-operators\xspace}
\nc{\mrho}{{\bm{\varrho}}}
\nc{\bfk}{\mathbf{K}}
\nc{\invlim}{\displaystyle{\lim_{\longleftarrow}}\,}
\nc{\ot}{\otimes}

\nc{\eval}[1]{\Big|_{#1}}

\newcommand{\add}{\frka\frkd}
\newcommand{\B }{\mathfrak{B}}
\newcommand{\lon }{\,\rightarrow\,}
\newcommand{\be }{\begin{equation}}
\newcommand{\ee }{\end{equation}}

\newcommand{\g}{\mathfrak g}
\newcommand{\h}{\mathfrak h}

\nc{\RR}{\mathbb{R}}

\nc{\CC}{\mathbb{C}}

\newcommand{\huaB}{\mathcal{B}}

\newcommand{\huaG}{\mathcal{G}}

\newcommand{\huaI}{\mathcal{I}}

\newcommand{\frka}{\mathfrak a}

\newcommand{\frkd}{\mathfrak d}

\newcommand{\Courant}[1]{\left\llbracket  #1\right\rrbracket }

\newcommand{\br}[1]{   [ \cdot,    \cdot  ]   }
\newcommand{\id}{\mathsf{id}}

\newcommand{\Ad}{\mathrm{Ad}}

\newcommand{\gl}{\mathfrak {gl}}

\newcommand{\ad}{\mathrm{ad}}

\nc{\CV}{\mathbf{C}}

\begin{document}

\title[Factorizable Lie bialgebras,   quadratic RB Lie algebras   and RB Lie bialgebras]{Factorizable Lie bialgebras,   quadratic Rota-Baxter Lie algebras   and Rota-Baxter Lie bialgebras}

\author{Honglei Lang}
\address{College of Science, China Agricultural University, Beijing, 100083, China
}
\email{hllang@cau.edu.cn}

\author{Yunhe Sheng}
\address{Department of Mathematics, Jilin University, Changchun 130012, Jilin, China}
\email{shengyh@jlu.edu.cn}


\begin{abstract}
In this paper, first we introduce the notion of quadratic Rota-Baxter Lie algebras of arbitrary weight, and show that there is a one-to-one correspondence between factorizable Lie bialgebras and   quadratic Rota-Baxter Lie algebras of nonzero weight. Then we introduce the notions of matched pairs, bialgebras and Manin triples of Rota-Baxter Lie algebras of arbitrary weight, and show that Rota-Baxter Lie bialgebras,  Manin triples of Rota-Baxter Lie algebras and certain matched pairs  of Rota-Baxter Lie algebras are equivalent. The   coadjoint representations  and quadratic Rota-Baxter Lie algebras play  important roles in the whole study.  Finally we generalize some results to the Lie group context. In particular, we show that there is a one-to-one correspondence between factorizable Poisson Lie groups and   quadratic Rota-Baxter Lie groups.
\end{abstract}


\keywords{factorizable Lie bialgebra,  quadratic Rota-Baxter Lie algebra, Rota-Baxter Lie bialgebra, factorizable Poisson Lie group\\
\quad  2020 \emph{Mathematics Subject Classification.}  17B38,  17B62,   22E60, 53D17}

\maketitle

\tableofcontents

\section{Introduction}

Poisson Lie groups, as semi-classical counterparts of Hopf algebras, or quantum groups, are Lie groups with compatible Poisson structures in the sense that the group multiplication is a Poisson map. The infinitesimals  of Poisson Lie groups are Lie bialgebras. A quasitriangular Hopf algebra \cite{D1} is a Hopf algebra $\mathcal{A}$ with an invertible element $R\in \mathcal{A}\otimes \mathcal{A}$, called quantum $R$-matrices, satisfying some conditions. Particularly, $R$ satisfies the quantum Yang-Baxter equation. The quasiclassical limits of quantum $R$-matrices are  classical $r$-matrices (\cite{STS}), which give rise to  the notion of quasitriangular Lie bialgebras.
Factorizable Lie bialgebras and factorizable Hopf algebras, introduced in \cite{RS},
are important classes of quasitriangular Lie bialgebras and quasitriangular Hopf algebras. For instance, the double of an arbitrary Lie bialgebra is a factorizable Lie bialgebra.
As the name suggests, factorizable Lie bialgebras (Hopf algebras) are used to
 establish the relation between classical $r$-matrices (quantum $R$-matrices) and certain factorization problems in Lie algebras (Hopf algebras).  A key feature of  the factorizable case  is that the symmetric part  of  a classical $r$-matrix $r\in\otimes ^2\g$ (quantum $R$-matrix) identifies the underlying vector space of the Lie algebra $\g$ (the Hopf algebra) with its dual space $\g^*$. Then a Lie bialgebra $(\g,\g^*)$ gives two compatible Lie algebra structures  on $\g$, which leads to a double Lie algebra. Such double Lie algebras generate Hamiltonian systems and can be solved by the factorization; see \cite{STS,S2}.

Rota-Baxter associative algebras were introduced in the probability study of G. Baxter. They have important applications in various areas, including the Connes-Kreimer's algebraic approach to renormalization of quantum field theory \cite{CK}, noncommutative
symmetric functions \cite{Fard,Yu-Guo}, splitting of
operads~\cite{BBGN,PBG}, quantum analogue of Poisson geometry \cite{Uchino08} and double Poisson algebras \cite{ARR,Goncharov,Schedler}. A linear map $B:\g\lon\g$ is called a {\bf Rota-Baxter operator of weight $\lambda$} on a Lie algebra $\g$ if
\begin{equation*}
[B(x), B(y)]_\g=B([B(x),y]_\g+[x,B(y)]_\g+\lambda[x,y]_\g), \quad \forall x,y\in\g.
\end{equation*} There is a close relationship between Rota-Baxter operators on Lie algebras and solutions of the (modified) classical Yang-Baxter equation. More precisely, a Rota-Baxter operator of weight 0 on a Lie algebra is naturally the operator form of a classical $r$-matrix~\cite{STS} under certain conditions;   Rota-Baxter operators of weight 1 on a Lie algebra one-to-one correspond to solutions of the modified classical Yang-Baxter equation.
To better understand such connections, the notion of an $\calo$-operator   (also called a relative Rota-Baxter operator \cite{PBG}
 or a generalized Rota-Baxter operator \cite{Uch})
on a Lie algebra $\g$ with respect to a representation $(V,\rho)$ was introduced by Kupershmidt in \cite{Ku},
which can be traced back to Bordemann's earlier work \cite{Bor}. An $\calo$-operator on a Lie algebra $\g$ with respect to a representation $(V,\rho)$ is a linear operator $T:V\to \g$ satisfying
$$
[Tu,Tv]_\g=T(\rho(Tu)v-\rho(Tv)u),\quad\forall u,v\in V.
$$
A skew-symmetric classical $r$-matrix $r\in\wedge^2\g^*$ naturally gives rise to an $\calo$-operator $r^\sharp:\g^*\to\g$ on $\g$ with respect to the coadjoint representation. On the other hand, any $\calo$-operator $T:V\to\g$ gives rise to a skew-symmetric solution of the classical Yang-Baxter equation in the semidirect product Lie algebra $\g\ltimes_\rho V$, see  \cite{Bai}. Rota-Baxter operators and $\calo$-operators also play important roles in the study of integrable systems \cite{BGN2010,STS}. For further details on
Rota-Baxter operators, see~\cite{Gub}.

 Even though the operator forms of skew-symmetric solutions of the classical Yang-Baxter equation are very clear, namely $\calo$-operators with respect to the coadjoint representation (equivalently Rota-Baxter operators of weight zero \cite{BD,STS} under certain conditions),  the operator forms of non-skew-symmetric solutions of the classical Yang-Baxter equation are still mysterious. Recently, Goncharov established a correspondence between non-skew-symmetric solutions of the classical Yang-Baxter equation and  Rota-Baxter operators of nonzero weight for certain Lie algebras \cite{G1,G2}. On the other hand, Kosmann-Schwarzbach clarified  the relation between factorizable Lie bialgebras and double Lie algebras \cite{K}, while the latter are closely related  to Rota-Baxter Lie algebras of weight 1. Thus it is natural to expect a direct relation between  factorizable Lie bialgebras and Rota-Baxter Lie algebras. This serves as the first purpose of the paper. We introduce the notion of a quadratic Rota-Baxter Lie algebra of weight $\lambda$, and establish a one-to-one correspondence between factorizable Lie bialgebras and quadratic Rota-Baxter Lie algebras of nonzero weight.

It is well known that quadratic Lie algebras play important roles in the theories of Manin triples of Lie algebras and Lie bialgebras. Since we have defined quadratic Rota-Baxter Lie algebras of arbitrary weight, it is natural to develop the bialgebra theory, the Manin triple theory as well as the matched pair theory for Rota-Baxter Lie algebras. To define matched pairs, we introduce the notion of representations of Rota-Baxter Lie algebras of weight $\lambda$. A Rota-Baxter Lie algebra of weight $\lambda$ admits the adjoint representation, as well as the coadjoint representation, while the latter is quite tricky. For a quadratic Rota-Baxter Lie algebra of weight $\lambda$, its adjoint representation and   coadjoint representation are isomorphic. The data in the coadjoint representation also serves as a guidance on how to define a  Rota-Baxter Lie bialgebra of weight $\lambda$. We introduce the notions of matched pairs, bialgebras and Manin triples of Rota-Baxter Lie algebras of weight $\lambda$ and show that they are equivalent.

Recently,  Rota-Baxter operators are defined in the categories of Lie groups \cite{GLS} and cocommutative Hopf algebras \cite{G}, and have found some applications \cite{BaG,BaG2}. Taking the differentiation, a Rota-Baxter operator on a Lie group gives a Rota-Baxter operator   on the corresponding Lie algebra. A Rota-Baxter operator on a cocommutative Hopf algebra gives rise to a Rota-Baxter operator on the Lie algebra of primitive elements and a Rota-Baxter operator on the group of group-like elements. We further generalize some results in the algebraic setting to the Lie group context. In particular, we show that factorizable Poisson Lie groups are in one-to-one correspondence with quadratic Rota-Baxter Lie groups, and a Rota-Baxter Lie group gives a matched pair of Lie groups. It is still mysterious to combine a Rota-Baxter operator and a Poisson structure on a Lie group to define a Rota-Baxter   Poisson   Lie group, and we will study it in the future. Another topic that we would like to explore in the future is the operator forms of solutions of the classical dynamical Yang-Baxter equation, which were first considered in \cite{F} and further studied in \cite{EV}.

The paper is organized as follows. In Section \ref{sec:fac}, we introduce the notion of a quadratic Rota-Baxter Lie algebra  of   weight $\lambda$, and show that   factorizable Lie bialgebras one-to-one correspond to quadratic Rota-Baxter Lie algebras  of nonzero weight (Theorem \ref{FL} and Theorem \ref{converse}). In Section \ref{sec:mp}, first we show that a Rota-Baxter Lie algebra gives a matched pair of Lie algebras. Then we introduce the notion of a representation of a Rota-Baxter Lie algebra of weight $\lambda$. A Rota-Baxter Lie algebra has the adjoint representation and the coadjoint representation, and for a quadratic Rota-Baxter Lie algebra they are isomorphic (Theorem \ref{thm:coadjoint}). Finally we introduce the notion of a matched pair of Rota-Baxter Lie algebras and show that it gives rise to a descendent matched pair of Lie algebras (Theorem \ref{demp}). In Section \ref{sec:rbb}, first we introduce the notion of a Rota-Baxter Lie bialgebra of weight $\lambda$, and show that they are equivalent to certain matched pairs of Rota-Baxter Lie algebras (Theorem \ref{bimp}). The double of a Rota-Baxter Lie bialgebra is still a   Rota-Baxter Lie bialgebra (Proposition \ref{abc}). We show that factorizable Lie bialgebras naturally give Rota-Baxter Lie bialgebras (Theorem \ref{pro:FL}). Then we introduce the notion of Manin triples of Rota-Baxter Lie algebras and show that they are equivalent to Rota-Baxter Lie bialgebras. In Section \ref{sec:rbg}, we show that factorizable Poisson Lie groups   one-to-one correspond to quadratic Rota-Baxter Lie groups (Theorem \ref{GFL} and Theorem \ref{thm:rbg-fac}), and a Rota-Baxter Lie group gives rise to a matched pair of Lie groups (Theorem \ref{thm:mp}).

\vspace{2mm}
\noindent
{\bf Acknowledgements. } This research is supported by NSFC (11922110, 11901568). We give our warmest thanks to Maxim Goncharov for helpful comments.

\section{Factorizable Lie bialgebras and quadratic Rota-Baxter Lie  algebras}\label{sec:fac}

In this section, we establish a one-to-one correspondence between factorizable Lie bialgebras and quadratic Rota-Baxter Lie  algebras of nonzero weight.

\subsection{Preliminaries on factorizable Lie bialgebras}
We first briefly recall the definitions of quasitriangular Lie bialgebras and factorizable Lie bialgebras following \cite{D1, EK, K,Lu, Lu2, RS, WX}.

\begin{defi}
A {\bf Lie bialgebra}  is a pair of Lie algebras $(\g,[\cdot,\cdot]_\g)$ and $(\g^*,[\cdot,\cdot]_{\g^*})$ such that
\[\Delta [x,y]_\g=[\Delta(x),y]_\g+[x,\Delta(y)]_\g,\quad \forall x,y\in \g,\]
where $\Delta:\g\to \wedge^2 \g$ is defined by $\langle \Delta(x),\xi\wedge \eta\rangle:=\langle x,[\xi,\eta]_{\g^*}\rangle$, and called the cobracket.
\end{defi}
A Lie bialgebra is denoted by $(\g,\g^*)$ or $(\g,\Delta)$.  For such a Lie bialgebra, there is a Lie algebra structure $[\cdot,\cdot]_{\bowtie}$ on the double $\g\oplus \g^*$ such that $\g$ and $\g^*$ are Lie subalgebras and
\[[x,\xi]_{\bowtie }=-\add_{\xi}^*x+\ad_x^*\xi,\quad \forall x\in \g,\xi\in \g^*,\]
where $\add_{\xi}^*x\in \g$ and $\ad_x^*\xi\in \g^*$ are given by
\[\langle \add_{\xi}^*x, \eta\rangle=-\langle x,[\xi,\eta]_{\g^*}\rangle,\quad \langle \ad_x^*\xi,y\rangle =-\langle \xi,[x,y]_{\g}\rangle.\]
Denote this Lie algebra by $\g\bowtie \g^*$. We refer the readers to \cite{Lu, Lu2} for the equivalent description of Lie bialgebras by using Manin triples of Lie algebras. One important observation is that if $(\g,\g^*)$ is a Lie bialgebra, then $(\g^*,\g)$ is also a Lie bialgebra.




Let $(\g,[\cdot,\cdot]_\g)$ be a Lie algebra. An element $r=\sum_i x_i\otimes y_i\in \g\otimes \g$  is called a solution of the classical Yang-Baxter equation if it satisfies
\[[r_{12},r_{13}]+[r_{13},r_{23}]+[r_{12},r_{23}]=0,\]
where $r_{12}=\sum_i x_i\otimes y_i\otimes 1$,   $r_{23}=\sum_i 1\otimes x_i\otimes y_i$ and $r_{13}=\sum_i x_i\otimes 1\otimes y_i$, and the bracket $[r_{12},r_{13}]$ is defined by
\[[r_{12},r_{13}]=[\sum_i x_i\otimes y_i\otimes 1,\sum_j x_j\otimes 1\otimes y_j]=\sum_{i,j}[x_i,x_j]_\g\otimes y_i\otimes y_j, \]
and similarly for $[r_{13},r_{23}]$ and $[r_{12},r_{23}]$.

Any $r\in \g\otimes \g$ induces a linear operator $r_+: \g^*\to \g$ defined by $r_+( \xi)= r(\xi,\cdot)$ for all $\xi\in\g^*.$  Define $r_-:=-r_+^*$.
 Let us introduce a bracket $[\cdot,\cdot]_r$ on $\g^*$:
\[[\xi,\eta]_r=\ad_{r_+ \xi}^* \eta-\ad^*_{r_- \eta} \xi,\quad \forall \xi,\eta\in \g^*.\]
If $r$ satisfies the $\ad$-invariant condition
\begin{equation}\label{eq:invr}
[r+\sigma(r),x]_\g=0,\quad \forall x\in\g,
\end{equation}
 where $\sigma$ is the flip operator which interchanges the components in $\g\otimes \g$, and it is a solution of the classical Yang-Baxter equation,  then $(\g^*,[\cdot,\cdot]_r)$ is a Lie algebra, which is denoted by $\g_r^*$. Moreover, $(\g,\g^*_r)$ constitutes a Lie bialgebra, which is called   a {\bf quasitriangular Lie bialgebra}. If $r$ is skew-symmetric, it is called a {\bf triangular Lie bialgebra}.

\begin{thm}{\rm(\cite{K,RS})}\label{thm:rcybe}
  Let  $r\in \g\otimes \g$ satisfy the $\ad$-invariant condition \eqref{eq:invr}. Then $r$ satisfies the classical Yang-Baxter equation  if and only if $(\g^*,[\cdot,\cdot]_r)$ is a Lie algebra and  the linear maps $r_+,r_-:(\g^*,[\cdot,\cdot]_r)\to (\g,[\cdot,\cdot]_\g)$ are both Lie algebra homomorphisms.
 \end{thm}

Denote by $I$ the operator
\begin{eqnarray}\label{I}
I=r_+-r_-: \g^*\to \g.
\end{eqnarray}
Note that $I^*=I$, and the $\ad$-invariant condition \eqref{eq:invr} is equivalent to
\begin{equation}\label{eq:invI}
  I\circ \ad_x^*=\ad_x\circ I,\quad \forall x\in \g.
\end{equation} Actually $\frac{1}{2}I$ is the symmetric part of $r$.
If $r$ is skew-symmetric, then $I=0$. Factorizable Lie bialgebras  are however  concerned with the opposite case that $I$ is nondegenerate; see \cite{RS}.
\begin{defi}
A quasitriangular Lie bialgebra $(\g, \g^*_r)$ defined by $r\in \g\otimes \g$ is called {\bf factorizable} if the linear map $I: \g^*\to \g$ defined in \eqref{I} is a linear isomorphism of vector spaces.
\end{defi}

\subsection{Factorizable Lie bialgebras and quadratic Rota-Baxter Lie  algebras}

Let $(\g, [\cdot,\cdot]_\g)$ be a Lie algebra  and $B:\g\lon\g$   a   Rota-Baxter operator of weight $\lambda$ on   $\g$. Then there is a new Lie bracket $[\cdot,\cdot]_B$ on $\g$ defined by
$$
[x,y]_B=[B(x),y]_\g+[x,B(y)]_\g+\lambda[x,y]_\g.
$$
The Lie algebra $(\g,[\cdot,\cdot]_B)$ is called the {\bf descendent Lie algebra}, and denoted by $\g_B$. It is obvious that $B$ is a Lie algebra homomorphism from $\g_B$ to $\g$:
$$
B[x,y]_B=[B(x),B(y)]_\g.
$$

Recall that a nondegenerate symmetric bilinear form $S\in \otimes ^2\g^*$ on a Lie algebra $\g$ is said to be invariant if
\begin{eqnarray}
\label{RBmanin1}S([x,y]_\g,z)+S(y,[x,z]_\g)&=&0,\quad \forall x,y\in \g.
\end{eqnarray}
A quadratic Lie algebra $(\g,S)$ is a Lie algebra $\g$ equipped with a nondegenerate symmetric invariant bilinear form $S\in \otimes ^2\g^*$.
\begin{defi}\label{defi:qua}
Let $(\g,[\cdot,\cdot]_\g,B)$ be a  Rota-Baxter Lie algebra of weight $\lambda$, and $S\in \otimes ^2\g^*$ a nondegenerate symmetric bilinear form.   The triple  $(\g,B,S)$ is called a {\bf quadratic Rota-Baxter Lie algebra of weight $\lambda$} if $(\g,S)$ is a quadratic Lie algebra and the following compatibility condition holds:
\begin{eqnarray}
\label{RBmanin}
S( x, {B} y)+S({B}x, y)+\lambda S(x,y)&=&0,\quad \forall x,y\in \g.
\end{eqnarray}

\end{defi}

A factorizable Lie bialgebra naturally gives rise to a quadratic Rota-Baxter Lie algebra of nonzero weight.

\begin{thm}\label{FL}
Let $(\g,\g^*_r)$ be a factorizable Lie bialgebra with $I=r_+-r_-$. Then $(\g,B,S_I)$ is a quadratic Rota-Baxter Lie algebra of weight $\lambda$, where the linear map $B:\g\to\g$ and $S_I\in\otimes ^2\g^*$ are defined by
\begin{eqnarray}\label{eq:defiB}
B&=&\lambda r_-\circ I^{-1},\\
 \label{eq:defiS_I}S_I(x,y)&=&\langle I^{-1}x,y\rangle,\quad \forall x,y\in\g.
\end{eqnarray}

\end{thm}

\begin{proof}
Since $r_+,r_-$ are both Lie algebra homomorphisms, for $x,y\in \g$, we have
\begin{eqnarray}\label{Ihomo}
I([I^{-1}x,I^{-1} y]_r)&=&(r_+-r_-)[I^{-1}x, I^{-1}y]_r\\ &=&\nonumber
[(I+r_-) I^{-1} x,(I+r_-)I^{-1} y]_\g-[r_-I^{-1}x,r_-I^{-1} y]_\g\\ &=&\nonumber
[r_-I^{-1}x, y]_\g+[x,r_-I^{-1}y]_\g+[x,y]_\g.
\end{eqnarray}
Therefore, we have
\begin{eqnarray*}
  B([Bx,y]_\g+[x,By]_\g+\lambda [x,y]_\g)&=&\lambda^2r_-I^{-1}\big([r_-I^{-1}x,y]_\g+[x,r_-I^{-1}y]_\g+[x,y]_\g\big)\\
  &=&\lambda^2 r_-[I^{-1}x,I^{-1}y]_r\\
  &=&\lambda^2 [r_-I^{-1} x,r_-I^{-1}y]_\g\\
  &=&[Bx,By]_\g,
\end{eqnarray*}
which implies  that $B$
 is a Rota-Baxter operator of weight $\lambda$ on $\g$.

Then, we show that $(\g,B,S_I)$ is a quadratic Rota-Baxter Lie algebra. Since $I^*=I$, it is obvious that $S_I\in \otimes^2 \g^*$ is symmetric.  It suffices to check \eqref{RBmanin1} and \eqref{RBmanin}. In fact, by  the invariant property \eqref{eq:invI}, we have $I^{-1}\circ \ad_x=\ad_x^*\circ  I^{-1}$. So
\begin{eqnarray*}
S_I([x,y]_\g,z)+S_I(y,[x,z]_\g)&=&\langle I^{-1}[x,y]_\g,z\rangle+\langle I^{-1} y,[x,z]_\g\rangle\\ &=&\langle I^{-1}\circ \ad_x(y)-\ad_x^*\circ I^{-1}(y),z\rangle\\
&=&0.
\end{eqnarray*}
Moreover, by using $r_-^*=-r_+$ and $I=r_+-r_-$, we have
\begin{eqnarray*}
S_I(x,By)+S_I(Bx,y)+\lambda S_I(x,y)&=&\lambda(\langle I^{-1}x, r_-\circ I^{-1}(y)\rangle+\langle I^{-1}\circ r_-\circ I^{-1}(x),y\rangle+\langle I^{-1}x,y\rangle)\\ &=&\lambda \langle (-I^{-1}\circ r_+\circ I^{-1}  +I^{-1}\circ r_-\circ I^{-1}+I^{-1})(x),y\rangle\\
&=&0.
\end{eqnarray*}
Therefore, $(\g,B,S_I)$ is a quadratic Rota-Baxter Lie algebra of weight $\lambda$.
\end{proof}

\begin{rmk}\label{relationwithG}
Let $(\g, B,S)$  be a quadratic Rota-Baxter Lie algebra of weight $\lambda$. Define   a linear operator $B^*:\g\to \g$ by
$S(x,By)=S(B^*x,y)$ for all $x,y\in \g$. Then   \eqref{RBmanin} is equivalent to $B+B^*=-\lambda \id$. Then by Theorem \ref{FL}, a factorizable Lie bialgebra $(\g,\g_r^*)$ gives rise to a Rota-Baxter operator $B$ on $\g$ such that $B+B^*=-\lambda \id$. Similar  results for simple Lie algebras and quadratic Lie algebras were given earlier by Goncharov; see \cite[Theorem 4]{G1} and \cite[Theorems 1, 3]{G2}.
\end{rmk}

It is well-known that if $B: \g\to \g$ is a Rota-Baxter operator of weight $\lambda$ on a Lie algebra, then
\begin{equation}\label{eq:Bt}
 \widetilde{B}:= -\lambda \id-B
\end{equation}  is also a Rota-Baxter operator of weight $\lambda$. So as a consequence of Theorem \ref{FL}, we have the following result.

\begin{cor}
 Let $(\g,\g^*_r)$ be a factorizable Lie bialgebra with $I=r_+-r_-$.  Then $(\g,\widetilde{B},S_I)$ is also a quadratic Rota-Baxter algebra of weight $\lambda$, where $\widetilde{B}=-\lambda\id-B=-\lambda r_+\circ I^{-1} $ and $S_I\in \otimes^2 \g^*$ is defined in \eqref{eq:defiS_I}.

\end{cor}
 \begin{proof}
It is obvious that $\widetilde{B}$ is a Rota-Baxter operator of weight $\lambda$. Since $B$ satisfies the invariant condition  \eqref{RBmanin}, we have
\begin{eqnarray*}
S_I(x,\widetilde{B}y)+S_I(\widetilde{B}x,y)+\lambda S_I(x,y)&=& -S_I(x, {B}y)-\lambda S_I(x,y)-S_I({B}x,y) -\lambda S_I(x,y)  +\lambda S_I(x,y)\\
&=&-S_I(x, {B}y) -S_I({B}x,y) -\lambda S_I(x,y)\\
&=&0,
\end{eqnarray*}
which implies that $(\g,\widetilde{B},S_I)$ is also a quadratic Rota-Baxter algebra of weight $\lambda$.
\end{proof}
\begin{cor}\label{liebi}
 Let $(\g,\g^*_r)$ be a factorizable Lie bialgebra with $I=r_+-r_-$, and $B$ the induced  Rota-Baxter operator of weight $\lambda$ on $\g$. Then   $\big(\g_B, (\g^*,[\cdot,\cdot]_{I})\big)$ is a  Lie bialgebra,  where
\[[\xi,\eta]_{I}:=\lambda I^{-1}([\frac{1}{\lambda} I\xi,\frac{1}{\lambda}I\eta]_\g),\quad \forall \xi,\eta\in \g^*,~ \lambda\neq 0.\]
Moreover, $\frac{1}{\lambda} I: \g^*\to \g$ gives a Lie bialgebra isomorphism from $\big(\g^*_r, \g \big)$ to $ \big(\g_B, (\g^*,[\cdot,\cdot]_{I})\big)$.
 \end{cor}
\begin{proof}
First we check that $\frac{1}{\lambda} I: \g_r^*\to \g_B$ is a Lie algebra isomorphism. In fact, for any $\xi,\eta\in \g^*$, taking $x=I \xi\in \g$ and $y=I \eta\in \g$, the equation \eqref{Ihomo} tells us that
\begin{eqnarray*}
\frac{1}{\lambda} I[\xi,\eta]_r=\frac{1}{\lambda^2} \big([BI\xi,I\eta]_\g+[I\xi,BI\eta]_\g+\lambda [I\xi,I\eta]_\g\big)=[\frac{1}{\lambda} I\xi,\frac{1}{\lambda} I\eta]_B.
\end{eqnarray*}
So $\frac{1}{\lambda}I$ is a Lie algebra isomorphism.

It is obvious that the map $(\frac{1}{\lambda} I)^*=\frac{1}{\lambda} I: (\g^*,[\cdot,\cdot]_I)\to \g$ is also Lie algebra isomorphism. 
Since $(\g_r^*,\g)$ is a Lie bialgebra, it follows that the pair $(\g_B,(\g^*,[\cdot,\cdot]_I))$ is also a Lie bialgebra and $\frac{1}{\lambda}I:\g_r^*\to \g_B$ is a Lie bialgebra isomorphism.
\end{proof}



\begin{ex}\label{doubleex}
Let $(\g,\g^*)$ be an arbitrary Lie bialgebra. Its Drinfeld double $\mathfrak{d}=\g\bowtie \g^*$ with $r=\sum_{i}\xi_i\otimes x_i\in \mathfrak{d}\otimes \mathfrak{d} $ is a quasitriangular Lie bialgebra (\cite{RS}), where $\{x_i\}$ is a basis of $\g$ and $\{\xi_i\}$ is its dual basis for $\g^*$.
Then $r_+,r_-: \mathfrak{d}^*\to \mathfrak{d}$ are given by
\[r_+(x,\xi)=(x,0),\quad r_-(x,\xi)=(0,-\xi),\quad \forall x\in \g,\xi\in \g^*.\]
Observe that  $I(x,\xi)=(x,\xi)$. So $(\mathfrak{d},\mathfrak{d}^*_r)$ is a factorizable Lie bialgebra.  By Theorem \ref{FL}, there associates a quadratic Rota-Baxter Lie algebra  $(\mathfrak{d},B,S_I)$, where
\begin{eqnarray*}
  B(x,\xi)&=&\lambda r_-\circ I^{-1}(x,\xi)=-\lambda (0,\xi),\\
   S_I(x+\xi,y+\eta)&=&\langle I^{-1}(x+\xi),y+\eta\rangle=\xi(y)+\eta(x).
\end{eqnarray*}

In this case, direct calculation shows that $\mathfrak{d}^*=\g\oplus \overline{\g^*}$, namely, the direct sum of the Lie algebra   $\g$ and the Lie algebra $\overline{\g^*}$, where $\overline{\g^*}$ denotes the Lie algebra  with the Lie bracket $-[\cdot,\cdot]_{\g^*}$.
\end{ex}

At the end of this section, we show that a quadratic Rota-Baxter Lie algebra of nonzero weight also gives rise to a factorizable Lie bialgebra.

\begin{thm}\label{converse}
Let $(\g, B,S)$ be a quadratic Rota-Baxter Lie algebra of weight $\lambda$ ($\lambda\neq 0)$, and $\huaI_S:\g^*\to\g$ the induced linear isomorphism given by $\langle \huaI_S^{-1} x,y\rangle:=S(x,y)$. Then  $r\in \g\otimes \g$ determined by
\[r_+:=\frac{1}{\lambda}(B+\lambda \id)\circ \huaI_S:\g^*\to \g, \quad r_+(\xi)=r(\xi,\cdot),\quad \forall\xi\in \g^*\]
satisfies the classical Yang-Baxter equation, and gives rise to a  quasitriangular Lie bialgebra $(\g,\g_r^*)$, which is factorizable.
\end{thm}

\begin{proof} Since $S$ is symmetric, it follows that  $\huaI_S^*=\huaI_S$.
By the fact that $S(x,By)+S(Bx,y)+\lambda S(x,y)=0$, we have
\[\langle \huaI_S^{-1} x,By\rangle+\langle \huaI_S^{-1} \circ B(x),y\rangle+\lambda\langle \huaI_S^{-1} x,y\rangle=0,\]
which implies that $B^*\circ \huaI_S^{-1}+\huaI_S^{-1} \circ B+\lambda \huaI_S^{-1}=0$, and then
\[\huaI_S\circ B^*+B\circ \huaI_S+\lambda \huaI_S=0.\]
Hence
\[r_-:=-r_+^*=-\frac{1}{\lambda}(\huaI_S\circ B^*+\lambda \huaI_S)=\frac{1}{\lambda} B\circ \huaI_S,\]
and $\huaI_S=r_+-r_-$.  Define a bracket operation $[\cdot,\cdot]_r$ on $\g^*$ by
\[[\xi,\eta]_r=\ad_{r_+ \xi}^* \eta-\ad_{r_-\eta}^*\xi,\quad \forall \xi,\eta\in \g^*.\]
We first check that
\begin{eqnarray}\label{im}
\frac{1}{\lambda}\huaI_S[\xi,\eta]_r=[\frac{1}{\lambda} \huaI_S\xi,\frac{1}{\lambda} \huaI_S\eta]_B,
\end{eqnarray}
 which would indicate that $[\cdot,\cdot]_r$ is a Lie bracket and $\frac{1}{\lambda}\huaI_S$ is a Lie algebra isomorphism from $(\g^*,[\cdot,\cdot]_r)$ to $(\g,[\cdot,\cdot]_B)$. Actually, following from $S([x,y]_\g,z)+S(y,[x,z]_\g)=0$, we have
 $$\langle \huaI_S^{-1}\circ \ad_x (y),z\rangle-\langle \ad_x^*\circ \huaI_S^{-1} (y),z\rangle=0,$$
 which yields the invariance of $\huaI_S$: \[\huaI_S\circ \ad_x^*=\ad_x\circ \huaI_S.\] Then
 we have
 \begin{eqnarray*}
 \huaI_S[\xi,\eta]_r&=&\huaI_S(\ad_{r_+ \xi}^* \eta-\ad_{r_-\eta}^*\xi)
 \\ &=&[r_+\xi,\huaI_S \eta]_\g-[r_-\eta,\huaI_S\xi]_\g
 \\ &=&[r_+\xi, r_+\eta-r_-\eta]_\g-[r_-\eta,r_+\xi-r_-\xi]_\g
 \\ &=&[r_+\xi,r_+\eta]_\g-[r_-\xi,r_-\eta]_\g.
 \end{eqnarray*}
 On the other hand, we have
 \begin{eqnarray*}
 [\huaI_S\xi,\huaI_S\eta]_B&=&[B\huaI_S\xi, \huaI_S\eta]_\g+[\huaI_S\xi, B\huaI_S\eta]_\g+\lambda [\huaI_S\xi,\huaI_S\eta]_\g
 \\ &=&\lambda [r_-\xi,r_+\eta-r_-\eta]_\g+\lambda [r_+\xi-r_-\xi,r_-\eta]_\g+\lambda [r_+\xi-r_-\xi,r_+\eta-r_-\eta]_\g\\ &=&
\lambda\big([r_+\xi,r_+\eta]_\g-[r_-\xi,r_-\eta]_\g\big),
 \end{eqnarray*}
which implies that \eqref{im} holds.
Finally since  $B+\lambda \id, B: (\g,[\cdot,\cdot]_B)\to (\g,[\cdot,\cdot]_\g)$ are both Lie algebra homomorphisms,
we deduce that
\[r_+:=\frac{1}{\lambda}(B+\lambda \id)\circ \huaI_S~\mbox{and}~ r_-=\frac{1}{\lambda} B\circ \huaI_S:   (\g^*,[\cdot,\cdot]_r)\to   (\g,[\cdot,\cdot]_\g)\]
 are both Lie algebra homomorphisms. Therefore,  by Theorem \ref{thm:rcybe}, $(\g,\g_r^*)$ is a quasitriangular Lie bialgebra. Since $\huaI_S=r_+-r_-$ is an isomorphism,  the Lie bialgebra  $(\g,\g_r^*)$ is factorizable.
\end{proof}

\begin{rmk}
The relation between factorizable Lie bialgebras and double Lie algebras (\cite{STS}) is explained clearly in \cite{K}.
Double Lie algebras relate closely 
with Rota-Baxter Lie algebras.  Let $\g$ be a Lie algebra and $R: \g\to \g$ be a linear map. Define $[x,y]_R=[Rx,y]_\g+[x,Ry]_\g$.  If $[\cdot,\cdot]_R$ is also a Lie bracket, the pair $(\g,R)$ is called a double Lie algebra. In particular, if $R$ satisfies the following modified Yang-Baxter equation:
\[[Rx,Ry]_\g-R([Rx,y]_\g+[x,Ry]_\g)+  [x,y]_\g=0, \]
$[\cdot,\cdot]_R$ satisfies the Jacobi identity and $(\g,R)$ is  a   double Lie algebra.  Indeed, $B$ is a Rota-Baxter operator of weight $1$ on $\g$ if and only if $R=\id+2B:\g\to \g$ is a solution of the modified Yang-Baxter equation.
Here we relate directly  factorizable Lie bialgebras with Rota-Baxter Lie algebras of nonzero weight.
\end{rmk}

\section{Matched pairs of Rota-Baxter Lie algebras}\label{sec:mp}

In this section, first we show that a Rota-Baxter Lie algebra gives rise to a matched pair of Lie algebras. Then we introduce the notion of  matched pairs of Rota-Baxter  Lie algebras of weight $\lambda$ and show that a matched pair  of Rota-Baxter  Lie algebras induces a descendent matched pair of Lie algebras.

\subsection{Rota-Baxter Lie algebras and matched pairs of Lie algebras}

Matched pairs of Lie algebras are also known as  twilled Lie algebras \cite{KM} or double Lie algebras \cite{Lu2}.

A {\bf matched pair of Lie algebras} (\cite{Majid,18}) consists of a pair of Lie algebras  $(\g,\h)$, a  representation $\rho: \g\to\gl(\h)$ of $\g$ on $\h$ and a   representation $\mu: \h\to\gl(\g)$ of $\h$ on $\g$ such that
\begin{eqnarray}
\label{eq:mp1}\rho(x) [\xi,\eta]_{\h}&=&[\rho(x)\xi,\eta]_{\h}+[\xi,\rho(x) \eta]_{\h}+\rho\big((\mu(\eta)x\big)\xi-\rho\big(\mu(\xi)x\big) \eta,\\
\label{eq:mp2}\mu(\xi) [x, y]_{\g}&=&[\mu(\xi)x,y]_{\g}+[x,\mu(\xi) y]_{\g}+\mu\big(\rho(y) \xi\big)x-\mu\big(\rho(x)\xi\big)y,
\end{eqnarray}
 for all $x,y\in \g$ and $\xi,\eta\in \h$. We will denote a matched pair of Lie algebras by $(\g,\h;\rho,\mu)$, or simply by $(\g,\h)$.

The following alternative description of matched pairs of Lie algebras is well-known.
\begin{pro}\label{eqdefi}
Let $(\g,\h;\rho,\mu)$ be a matched pair of Lie algebras. Then there is a Lie algebra structure on the direct sum space $\g\oplus \h$ with the Lie bracket $[\cdot,\cdot]_{\bowtie }$ given by
\[[x+\xi,y+\eta]_{\bowtie }=\big([x,y]_\g+\mu(\xi)y-\mu(\eta) x\big)+\big([\xi,\eta]_\h+\rho(x)\eta-\rho(y) \xi).\]
Denote this Lie algebra by $\g\bowtie \h$.

Conversely, if $(\g\oplus \h,[\cdot,\cdot])$ is a Lie algebra such that $\g$ and $\h$ are Lie subalgebras, then  $(\g,\h;\rho,\mu)$ is a matched pair of Lie algebras, where the representations $\rho: \g\to\gl(\h)$  and $\mu: \h\to\gl(\g)$ are determined  by
\[[x,\xi] =\rho(x) \xi-\mu(\xi) x,\quad \forall x\in \g,\xi\in \h.\]
\end{pro}

Let $(\g,[\cdot,\cdot]_\g)$ be a Lie algebra. We denote by $(\g\oplus \g,[\cdot,\cdot]_{\g\oplus \g})$ the direct sum Lie algebra
with the bracket given by
\[[(x_1,x_2),(y_1,y_2)]_{\g\oplus \g}=([x_1,y_1]_\g,[x_2,y_2]_\g).\]

\begin{pro}\label{doublelie}
Let $(\g,[\cdot,\cdot]_\g,B)$ be a Rota-Baxter Lie algebra of weight $\lambda$ ($\lambda\neq 0$) with the descendent Lie algebra  $\g_B$.  Then there is a Lie algebra structure $[\cdot,\cdot]_D$ on the vector space $\g\oplus \g$ given by
\begin{eqnarray*}
[(0,x),(0,y)]_D&=&(0,[x,y]_\g);\\
{[(\xi,0),(\eta,0)]_D}&=&([\xi,\eta]_B,0);\\
{[(0,x), (\xi,0)]_D}&=&-[(\xi,0),(0,x)]_D=([x,\xi]_\g, [x,B\xi]_\g-B[x,\xi]_\g).
\end{eqnarray*}
 Moreover, there is a Lie algebra isomorphism from $(\g\oplus \g,[\cdot,\cdot]_D )$  to the direct sum Lie algebra $\g\oplus \g$ defined by
\begin{eqnarray}\label{phi}
\phi:\g\oplus \g\to \g\oplus \g,\quad (\xi,x)\mapsto (B\xi+\lambda\xi+x,B\xi+x).
\end{eqnarray}
\end{pro}
\begin{proof}
First, it is direct to show that $\phi$ is an isomorphism between vector spaces whose inverse is
\[\phi^{-1}:\g\oplus \g\to \g\oplus \g,\qquad (x,y)\mapsto \frac{1}{\lambda}(x-y, \lambda y-B(x-y)).\]
In the sequel, we show that the bracket $[\cdot,\cdot]_D$ is exactly the pull-back of the  direct sum Lie algebra structure on $\g\oplus \g$ by $\phi$, i.e.
$$
[e_1,e_2]_D=\phi^{-1}[\phi(e_1),\phi(e_2)]_{\g\oplus \g},\quad \forall e_1, e_2\in \g\oplus\g.
$$
Consequently, $(\g\oplus \g,[\cdot,\cdot]_D )$ is a Lie algebra and $\phi$ is a Lie algebra isomorphism.

 By the fact that $B[\xi,\eta]_B=[B\xi,B\eta]_\g$, we have\begin{eqnarray*}
\phi^{-1}[\phi(\xi,0),\phi(\eta,0)]_{\g \oplus \g}&=&\phi^{-1}[(B\xi+\lambda \xi,B\xi),(B\eta+\lambda \eta,B\eta)]_{\g \oplus \g}\\
 &=&\phi^{-1}([B\xi+\lambda \xi,B\eta+\lambda \eta]_\g,[B\xi,B\eta]_\g)\\
 &=&\phi^{-1}(B[\xi,\eta]_B+\lambda [\xi,\eta]_B,B[\xi,\eta]_B)\\ &=& [(\xi,0),(\eta,0)]_D,
\end{eqnarray*}
and
\begin{eqnarray*}
&&\phi^{-1}[\phi(0,x),\phi(\xi,0)]_{\g\oplus \g}\\&=&\phi^{-1}[(x,x),(B\xi+\lambda \xi,B\xi)]_{\g\oplus \g}
\\ &=&\phi^{-1}([x,B\xi+\lambda\xi]_\g, [x,B\xi]_\g)
\\ &=&\phi^{-1}(B[x,\xi]_\g+\lambda [x,\xi]_\g+[x,B\xi]_\g-B[x,\xi]_\g, B[x,\xi]_\g+[x,B\xi]_\g-B[x,\xi]_\g)
\\ &=&([x,\xi]_\g,[x,B\xi]_\g-B[x,\xi]_\g)
\\&=&[(0,x),(\xi,0)]_D.
\end{eqnarray*}
It is obvious that $[(0,x),(0,y)]_D=\phi^{-1}[\phi(0,x),\phi(0,y)]_{\g\oplus \g}$. Therefore the bracket $[\cdot,\cdot]_D$ is exactly the pull-back of the  direct sum Lie bracket $[\cdot,\cdot]_{\g\oplus \g}$ by $\phi$.
\end{proof}
 Parallel results of Proposition \ref{doublelie} for associative algebras can be found in \cite{G3}.
\begin{thm}\label{thm:mpB}
 Let $(\g,[\cdot,\cdot]_\g,B)$ be a Rota-Baxter Lie algebra of weight $\lambda$ ($\lambda\neq 0$) with the descendent Lie algebra  $\g_B$.  Then $(\g_B,\g;\rho,\mu)$ is a matched pair of Lie algebras, where $\rho:\g_B\to\gl(\g)$ and $\mu:\g\to \gl(\g_B)$ are given by
\begin{equation}\label{eq:mpr}
 \rho(\xi)(x)= B[x,\xi]_\g -[x,B\xi]_\g,\quad \mu(x)(\xi)=  [x,\xi]_\g,\quad \forall x\in\g, \xi\in\g_B.
\end{equation}
Moreover, the corresponding Lie algebra $\g_B\bowtie \g$ is exactly the Lie algebra $(\g\oplus \g,[\cdot,\cdot]_D)$ given in Proposition \ref{doublelie}.
\end{thm}
\begin{proof}
By Proposition \ref{doublelie},  it is obvious that both $\g_B$ and $\g$ are Lie subalgebras of the Lie algebra $(\g\oplus \g,[\cdot,\cdot]_D )$, and
$$
{[(\xi,0),(0,x)]_D} =(\rho(\xi)(x), -\mu(x)(\xi)).
$$ Thus,  $(\g_B,\g;\rho,\mu)$ is a matched pair of Lie algebras. The other conclusion is obvious.
\end{proof}

Since $\phi$ is a Lie algebra isomorphism and $\g_B$ is a Lie subalgebra, it follows that $\im(\phi|_{\g_B})\subset \g\oplus \g$ is a Lie subalgebra, which is isomorphic to $\g_B$. It gives an alternative approach to  prove the factorization theorem of Rota-Baxter Lie algebras. We refer to \cite{STS, GLS} for more details for the factorization theorem of Rota-Baxter Lie algebras of weight $1$.
\begin{cor}
Let $(\g,[\cdot,\cdot]_\g,B)$ be a Rota-Baxter Lie algebra of weight $\lambda$ $(\lambda\neq 0)$. Then for any $x\in \g$, there exists a unique decomposition
$x=x_+-x_-$ with
 $(x_+,x_-)\in \im(\phi|_{\g_B})\subset \g\oplus \g$.
\end{cor}
\begin{proof}
We have
$x=\frac{1}{\lambda}(Bx+\lambda x)-\frac{1}{\lambda}Bx=x_+-x_-$, where $(x_+,x_-)=\phi(\frac{1}{\lambda} x,0)$. The decomposition is unique since $\phi: \g_B\to \im(\phi|_{\g_B})$ is a Lie algebra isomorphism.
\end{proof}

Let $(\g,\g^*_r)$ be a factorizable Lie bialgebra with $I=r_+-r_-$. Then $(\g,B=\lambda r_-\circ I^{-1})$ is a Rota-Baxter Lie algebra of weight $\lambda$. By Theorem  \ref{thm:mpB}, we have the Lie algebra $\g_B\bowtie \g$,  which is the double of the matched pair $(\g_B,\g;\rho,\mu)$. On the other hand, there is a Lie algebra $\g^*_r\bowtie \g$, which is the double of the Lie bialgebra $(\g^*_r,\g)$.

\begin{pro}\label{algebradiagram}
With the above notations, we have a commutative diagram of Lie algebra isomorphisms:
\begin{eqnarray}\label{diagram}
		\vcenter{\xymatrix{
			\g_r^*\bowtie \g \ar[d]_(0.45){(\frac{1}{\lambda}I)\oplus \id}\ar[r]^(0.45){\psi} &\g\oplus \g\\
			\g_B\bowtie \g,\ar[ur]_{\phi}&
		}}
\end{eqnarray}
where $\phi$ is defined by \eqref{phi} and $\psi$ is given by
\[\psi(\xi,x)=(r_+\xi+x,r_-\xi+x),\quad \forall x\in \g,\xi\in \g^*.\]
\end{pro}

\begin{proof}
We first show that $\tilde{I}:=(\frac{1}{\lambda} I)\oplus \id: \g_r^*\bowtie \g\to \g_B\bowtie \g$ is a Lie algebra isomorphism. We only need to check the relation
\begin{eqnarray}\label{tildeI}
\tilde{I}([x,\xi]_{\bowtie})=[\tilde{I}(x),\tilde{I}(\xi)]_D,\quad \forall x\in \g,\xi\in \g^*.
\end{eqnarray}
In fact, by the  $\ad$-invariant condition \eqref{eq:invI}, we have
\[\tilde{I}([x,\xi]_{\bowtie})=(\frac{1}{\lambda} I(\ad_x^*\xi),-\add_\xi^* x)=(\frac{1}{\lambda} [x,I\xi]_\g, -\add_\xi^*x ),\]
and
\[ {[\tilde{I}(x),\tilde{I}(\xi)]_D}=[(0,x),(\frac{1}{\lambda} I\xi,0)]_D=(\frac{1}{\lambda} [x,I\xi]_\g, \frac{1}{\lambda} [x,BI\xi]_\g-\frac{1}{\lambda} B[x,I\xi]_\g).\]
Taking pairing with any $\eta\in \g^*$ and since $r_-=-r_+^*$, one finds
\begin{eqnarray*}
\langle -\add_\xi^* x,\eta\rangle &=&-\langle x,\ad_{r_+\eta}^* \xi-\ad_{r_-\xi}^* \eta\rangle
 \\ &=&\langle [x,r_-\xi]_\g,\eta\rangle+\langle [r_+\eta, x]_\g,\xi\rangle\\ &=&
\langle \frac{1}{\lambda}[x,BI\xi]_\g-r_-(\ad_x^*\xi),\eta\rangle\\
&=&\langle \frac{1}{\lambda} [x,BI\xi]_\g-\frac{1}{\lambda}BI(\ad_x^*\xi),\eta\rangle\\ &=&\frac{1}{\lambda}\langle [x,BI\xi]_\g-B[x,I\xi]_\g,\eta\rangle,
\end{eqnarray*}
which implies that \eqref{tildeI} holds.

Then, for $(\xi,x)\in \g_r^*\bowtie \g$, we have
\[\phi\circ (\frac{1}{\lambda}I\oplus \id)(\xi,x)=\phi(\frac{1}{\lambda}I\xi,x)=(\frac{1}{\lambda} BI\xi+I\xi+x,\frac{1}{\lambda} BI\xi+x)=(r_+\xi+x,r_-\xi+x)=\psi(\xi,x).\]
Since both $\phi$ and $\frac{1}{\lambda} I\oplus \id$ are Lie algebra isomorphisms, it follows that $\psi$ is also a Lie algebra isomorphism.
\end{proof}

\subsection{Matched pairs of Rota-Baxter Lie algebras}

In this subsection, we introduce the notion of matched pairs of Rota-Baxter Lie algebras of weight $\lambda$, and show that a matched pair  of Rota-Baxter Lie algebras of weight $\lambda$ gives rise to  a descendent matched pair of Rota-Baxter Lie algebras.

Since in the definition of a matched pair of Lie algebras, we need the concept of representations of Lie algebras. For the purpose of defining matched pairs of Rota-Baxter Lie algebras, we introduce the notion of representations of Rota-Baxter Lie algebras of weight $\lambda$. Note that the concept of representations of Rota-Baxter Lie algebras of weight 0 was already given in \cite{JS} in the study of cohomologies of Rota-Baxter Lie algebras, and the notion of representations of Rota-Baxter associative algebras was introduced in \cite{GouLin}, and further studied in \cite{QP}.

\begin{defi}
A {\bf representation of a Rota-Baxter Lie algebra} $(\g,B)$ of weight $\lambda$ on a vector space $W$ with respect to a linear transformation $T\in\gl(W)$ is a representation $\rho$ of the Lie algebra $\g$ on $W$, satisfying
\begin{equation*}
\rho(Bx)(T u)=T\Big(\rho(Bx) u +   \rho(x)(T u)+\lambda \rho(x) u\Big), \quad \forall x\in\g, u\in W.
\end{equation*}
\end{defi}
We will denote a representation by $(W,T,\rho).$

 Let $(W,T,\rho)$ be a representation of a Rota-Baxter Lie algebra $(\g,B)$ of weight $\lambda$. Since $(W, \rho)$ is a representation of the Lie algebra $ \g $, we have the semidirect product Lie algebra $\g\ltimes W$. Then define the map
\[B\oplus T: \g\ltimes W\to \g\ltimes W,\qquad x+u \mapsto Bx+Tu.\]
The following conclusion is obvious.
\begin{pro}
 With  above notations, $(\g\ltimes W,B\oplus T)$ is a Rota-Baxter Lie algebra of weight $\lambda$, called the semidirect product of $(\g,B)$ and the representation $(W,T,\rho)$.
\end{pro}
\begin{ex}
  It is straightforward  to see that $(\g,B,\ad)$ is a representation of a Rota-Baxter Lie algebra  $(\g,B)$ of weight $\lambda$, which is called the {\bf adjoint representation}  of $(\g,B)$.
\end{ex}

\begin{defi}
 Let $(W,T,\rho)$ and $(W',T',\rho')$ be two representations of a  Rota-Baxter Lie algebra  $(\g,B)$ of weight $\lambda$. A homomorphism from $(W,T,\rho)$ to $(W',T',\rho')$ is a linear map $\phi:W\to W'$ such that
 \begin{eqnarray*}
   \phi\circ \rho(x)&=&\rho'(x)\circ \phi,\quad \forall x\in\g,\\
   \phi\circ T&=&T'\circ \phi.
 \end{eqnarray*}
\end{defi}

Rota-Baxter Lie algebras  of weight $\lambda$ not only admit adjoint representations, but also coadjoint representations. The following result plays an important role in our later study of Rota-Baxter Lie bialgebras  of weight $\lambda$ (see Theorem \ref{bimp}).

\begin{thm}\label{thm:coadjoint}
 Let $(\g,B)$ be a Rota-Baxter Lie algebra  of weight $\lambda$.  Then $(\g^*,-\lambda\id-B^*,\ad^*)$ is a representation, which is called the {\bf coadjoint representation} of $(\g,B)$.

 Moreover, if  $(\g,B,S)$ is a quadratic Rota-Baxter Lie algebra  of weight $\lambda$, then the linear map $S^\sharp:\g\to\g^*$ defined by
 $
 \langle S^\sharp(x),y\rangle=S(x,y)
 $
for all $x,y\in\g$, is an isomorphism  from the adjoint representation $(\g,B,\ad)$ to the coadjoint representation $(\g^*,-\lambda\id-B^*,\ad^*)$.
\end{thm}
\begin{proof}
  For all $\xi\in\g^*$ and $x,y\in\g$, since $B$ is a Rota-Baxter operator of weight $\lambda$ on $\g$, we have
  \begin{eqnarray*}
    &&\langle \ad^*_{Bx}(-\lambda\id-B^*)\xi-(-\lambda\id-B^*)\Big(\ad^*_{Bx}\xi+\ad^*_x(-\lambda\id-B^*)\xi+\lambda\ad^*_x\xi   \Big),y\rangle\\
    &=&\langle\xi, \lambda[Bx,y]_\g+B[Bx,y]_\g-\lambda[Bx,y]_\g-[Bx,By]_\g
    +\lambda B[x,y]_\g+B[x,By]_\g\rangle
    \\ &=&0,
      \end{eqnarray*}
      which implies that $(\g^*,-\lambda\id-B^*,\ad^*)$ is a representation.

      Let  $(\g,B,S)$ be a quadratic Rota-Baxter Lie algebra  of weight $\lambda$. By \eqref{RBmanin1}, we have
      $$
      S^\sharp\circ \ad_x=\ad_x^*\circ S^\sharp,
      $$
      By \eqref{RBmanin}, we have
      $$
      S^\sharp\circ B= (-\lambda\id-B^*)\circ S^\sharp.
      $$
      Therefore, $S^\sharp:\g\to\g^*$
 is a homomorphism from the adjoint representation $(\g,B,\ad)$ to the coadjoint representation $(\g^*,-\lambda\id-B^*,\ad^*)$.
\end{proof}

\begin{rmk}
  We emphasize that in general $(\g^*, B^*,\ad^*)$ is not a representation of $(\g,B)$. 
\end{rmk}

Now we introduce the notion of matched pairs of Rota-Baxter Lie algebras  of weight $\lambda$.

\begin{defi}\label{RBMP}
A {\bf matched pair of Rota-Baxter Lie algebras} of weight $\lambda$ consists of a pair of Rota-Baxter Lie algebras  $((\g,B),(\h,C))$ of weight $\lambda$, a  representation  $\rho: \g\to \gl(\h)$ of the Rota-Baxter Lie algebra $(\g,B)$ on $(\h,C)$ and a   representation $\mu: \h\to \gl(\g)$ of the Rota-Baxter Lie algebra $(\h,C)$ on $(\g,B)$ such that
  $(\g,\h; \rho,\mu)$ is a matched pair of Lie algebras.
\end{defi}

We will denote a matched pair of Rota-Baxter Lie algebras of weight $\lambda$ by $((\g,B),(\h,C);\rho,\mu)$, or simply by $((\g,B),(\h,C))$.

It is straightforward to obtain the following alternative characterization of matched pairs of Rota-Baxter Lie algebras.

\begin{pro}
 Let $(\g,B)$ and $(\h,C)$ be  Rota-Baxter Lie algebras  of weight $\lambda$, $\rho: \g\to \gl(\h)$ a  representation of the   Lie algebra $\g$ on $\h$ and $\mu: \h\to \gl(\g)$ a   representation of the  Lie algebra $\h$ on $\g$. Then $((\g,B),(\h,C);\rho,\mu)$ is a matched pair of Rota-Baxter Lie algebras if and only if $(\g,\h;\rho,\mu)$ is a matched pair of  Lie algebras and the following equalities hold:
 \begin{eqnarray}
\label{RBmp1}C(\rho(Bx)\xi+\rho(x)(C \xi)+\lambda \rho(x)\xi)&=&\rho(Bx) (C \xi);\\
\label{RBmp2} B(\mu(C\xi) x+\mu(\xi)(B x)+\lambda \mu(\xi)x)&=&\mu(C \xi) (Bx),
\end{eqnarray}
for all $x\in \g$ and $\xi\in \h$.
\end{pro}

On the double of a matched pair of Lie algebras, there is a Lie algebra structure. There is also a  Rota-Baxter Lie algebra structure on the double of a matched pair of Rota-Baxter Lie algebras as expected.
\begin{pro}\label{RBon}
Let   $(\g,B)$ and $(\h, C)$ be Rota-Baxter  Lie algebras  of weight $\lambda$, $\rho:\g\to \gl(\h)$ and $\mu:\h\to \gl(\g)$ are linear maps. Then $((\g,B),(\h,C);\rho,\mu)$ is a matched pair of Rota-Baxter Lie algebras  of weight $\lambda$   if and only if $(\g,\h;\rho,\mu)$ is a matched pair of  Lie algebras and
\[B\oplus C:\g\oplus \h\to \g\oplus \h,\qquad x+\xi\mapsto B x+C \xi,\]
is a Rota-Baxter operator of weight $\lambda$ on the Lie algebra $\g\bowtie \h$.

\end{pro}
\begin{proof}
Let $((\g,B),(\h,C);\rho,\mu)$ be a matched pair of Rota-Baxter Lie algebras  of weight $\lambda$. By definition, $(\g,\h;\rho,\mu)$ is a matched pair of  Lie algebras. By \eqref{RBmp1} and \eqref{RBmp2}, we have
\[[Bx,C \xi]_{\bowtie}=(B\oplus C)\big([B x,\xi]_{\bowtie}+[x,C \xi]_{\bowtie}+\lambda[x,\xi]_{\bowtie}\big).\]
Therefore $B\oplus C$ is a Rota-Baxter operator of weight $\lambda$ on the Lie algebra $\g\bowtie \h$.

The converse part can be proved similarly.
\end{proof}

Let $((\g,B),(\h,C);\rho,\mu)$ be a matched pair of Rota-Baxter Lie algebras  of weight $\lambda$. Then there are three descendent Lie algebras $\g_{B}, \h_{C}$ and $(\g\bowtie \h)_{B\oplus C}$ coming from the three Rota-Baxter operators $B:\g\to \g, C:\h\to \h$ and $B\oplus C:\g\oplus \h\to \g\oplus \h$ of weight $\lambda$, respectively. In the following theorem, their relation is clarified and a descendent matched pair of Lie algebras is obtained.

\begin{thm}\label{demp}
Let $((\g,B),(\h,C);\rho,\mu)$ be a matched pair of Rota-Baxter Lie algebras  of weight $\lambda$. Then    $(\g_{B},\h_{C};\rho_{(B,C)},\mu_{(B,C)})$ is a   matched pair of Lie algebras, where   $\rho_{(B,C)} $ and $\mu_{(B,C)}$ are  given by
\begin{eqnarray}
\label{eq:rep1}\rho_{(B,C)}(x) \xi&=&\rho(Bx) \xi+\rho(x)(C \xi)+\lambda\rho(x) \xi,\\
\label{eq:rep2}\mu_{(B,C)}(\xi) x&=&\mu(C\xi)x+\mu(\xi)(B x)+\lambda \mu(\xi) x.
\end{eqnarray}
Moreover, we have
\[\g_{B}\bowtie \h_{C}=(\g\bowtie \h)_{B\oplus C}\]
as Lie algebras.
\end{thm}
The matched pair $(\g_{B},\h_{C};\rho_{(B,C)},\mu_{(B,C)})$ is called the {\bf descendent matched pair}.
\begin{proof}
By Proposition \ref{RBon}, there is a descendent Lie algebra structure on $\g\oplus \h$, denoted by $(\g\bowtie \h)_{B\oplus C}$, which contains $\g_B$ and $\h_C$ as Lie subalgebras. Let us examine this Lie bracket on the crossing term. For $x\in \g$ and $\xi\in \h$, we have
\begin{eqnarray*}
[x,\xi]_{(\g\bowtie \h)_{B\oplus C}}&=&[Bx,\xi]_{\bowtie}+[x,C\xi]_{\bowtie}+\lambda[x,\xi]_{\bowtie}
\\ &=&\rho(Bx)\xi+\rho(x)(C \xi)+\lambda \rho(x) \xi-\big(\mu(C\xi) x+\mu(\xi)(B x)+\lambda \mu(\xi) x\big)\\
&=&\rho_{(B,C)} (x)\xi-\mu_{(B,C)}(\xi) x.
\end{eqnarray*}
Then by Proposition \ref{eqdefi}, $(\g_{B},\h_{C};\rho_{(B,C)},\mu_{(B,C)})$ forms a matched pair of Lie algebras. Moreover, the induced Lie algebra on its double  $\g_B\bowtie \h_C$ coincides with  $(\g\bowtie \h)_{B\oplus C}$.
\end{proof}

\begin{cor}
Let $((\g,B),(\h,C);\rho,\mu)$ be a matched pair of Rota-Baxter Lie algebras  of weight $\lambda$. Then    $((\g_{B},B),(\h_{C},C);\rho_{(B,C)},\mu_{(B,C)})$ is a   matched pair of Rota-Baxter  Lie algebras of weight $\lambda$.
\end{cor}
\begin{proof}
   Using the well-known fact that a Rota-Baxter operator of weight $\lambda$ on a Lie algebra is also a Rota-Baxter operator on the descendent Lie algebra, it follows that $B\oplus C$ is a  Rota-Baxter operator of weight $\lambda$ on $(\g\bowtie \h)_{B\oplus C}$. By Theorem \ref{demp}, $\g_{B}\bowtie \h_{C}=(\g\bowtie \h)_{B\oplus C}$. Therefore, $B\oplus C$ is a Rota-Baxter operator of weight $\lambda$ on $\g_{B}\bowtie \h_{C}$.    By Proposition \ref{RBon},  $((\g_{B},B),(\h_{C},C);\rho_{(B,C)},\mu_{(B,C)})$ is a   matched pair of Rota-Baxter  Lie algebras of weight $\lambda$.
  \end{proof}
\begin{ex}
Let $(\g,[\cdot,\cdot]_\g,B)$ and $(\h,[\cdot,\cdot]_\h,C)$ be two  Rota-Baxter Lie algebras of weight $\lambda$. Then $((\g,B),(\h,C);\rho=0,\mu=0)$ becomes a matched pair of Rota-Baxter Lie algebras.
\end{ex}

\begin{pro}
Let $(\g,B)$ be a Rota-Baxter Lie algebra of weight $\lambda$.  Then $((\g_B,B),(\g,B);\rho,\mu)$ is   a matched pair  of Rota-Baxter Lie algebras  of weight $\lambda$, where $\rho$ and $\mu$ are given by \eqref{eq:mpr}.
\end{pro}
\begin{proof}
By Proposition \ref{doublelie}, $(\g_B,\g;\rho,\mu)$ is   a matched pair  of  Lie algebras. As $B$ is a Rota-Baxter operator of weight $\lambda$ on both $\g$ and $\g_B$, it suffices to check the Conditions \eqref{RBmp1} and \eqref{RBmp2}. In fact, Condition \eqref{RBmp1} follows directly from the fact that $B$ is a Rota-Baxter operator of weight $\lambda$ on the Lie algebra  $\g$. In order to show Condition \eqref{RBmp2}, we first note that
\begin{eqnarray*}
&&B(\rho(B\xi) x+\rho(\xi)(Bx)+\lambda \rho(\xi)x)\\ &=&B(B[x,B\xi]_\g-[x,B^2\xi]_\g+B[Bx,\xi]_\g-[Bx,B\xi]_\g+\lambda B[x,\xi]_\g-\lambda[x,B\xi]_\g)\\ &=&
B(-[x,B^2\xi]_\g-\lambda[x,B\xi]_\g).
\end{eqnarray*}
On the other hand, as $B$ is a Rota-Baxter operator on $\g$, we have
\begin{eqnarray*}
\rho(B\xi)(Bx)&=&B[Bx,B\xi]_\g-[Bx,B^2\xi]_\g\\ &=&B([Bx,B\xi]_\g-[Bx,B\xi]_\g-[x,B^2\xi]_\g-\lambda[x,B\xi]_\g)\\ &=&
B(-[x,B^2\xi]_\g-\lambda[x,B\xi]_\g).
\end{eqnarray*}
So Condition \eqref{RBmp2} is obtained. Therefore, $((\g_B,B),(\g,B),\rho,\mu)$ is   a matched pair  of Rota-Baxter Lie algebras  of weight $\lambda$.
\end{proof}

\section{Rota-Baxter Lie bialgebras}\label{sec:rbb}

In this section, we introduce the notion of Rota-Baxter Lie bialgebras of weight $\lambda$, and show that Rota-Baxter Lie bialgebras, certain matched pairs of Rota-Baxter Lie  algebras and Manin triples of Rota-Baxter Lie algebras are equivalent.

\subsection{Rota-Baxter Lie bialgebras and matched pairs of Rota-Baxter Lie  algebras}


We introduce the following definition of Rota-Baxter operators on a Lie bialgebra and show that Rota-Baxter   Lie bialgebras and certain  matched pairs of Rota-Baxter Lie  algebras are equivalent.  In particular, we show that the Drinfeld double of a Rota-Baxter   Lie bialgebra is still a Rota-Baxter   Lie bialgebra, and factorizable Lie bialgebras also give rise to Rota-Baxter Lie bialgebras.
\begin{defi}
A {\bf Rota-Baxter operator of weight $\lambda$} on a Lie bialgebra $(\g,\g^*)$ is a linear map $B:\g\to \g$ such that
\begin{itemize}
\item[\rm (i)] $B$ is a Rota-Baxter operator of weight $\lambda$ on $\g$;
\item[\rm (ii)]  $\widetilde {B}^*:=-\lambda\id-B^*$ is a Rota-Baxter operator of weight $\lambda$ on $\g^*$. 
\end{itemize}
A Lie bialgebra with a Rota-Baxter operator of weight $\lambda$ is called a {\bf Rota-Baxter Lie bialgebra of weight $\lambda$}.
\end{defi}
We denote a Rota-Baxter Lie bialgebra of weight $\lambda$ by $(\g,\g^*,B)$.

As a well-known fact of Rota-Baxter operators,  $\widetilde {B}^*=-\lambda\id-B^*:\g^*\to \g^*$ is a Rota-Baxter operator of weight $\lambda$ on $\g^*$  if and only if $B^*:\g^*\to \g^*$ is a Rota-Baxter operator of weight $\lambda$. The descendent Lie brackets of $\widetilde {B}^*$ and $B^*$ on $\g^*$ are related by
\begin{eqnarray*}
[\xi,\eta]_{\widetilde {B}^*}=-[B^*\xi,\eta]_{\g^*}-[\xi,B^*\eta]_{\g^*}-\lambda[\xi,\eta]_{\g^*}=-[\xi,\eta]_{B^*}.
\end{eqnarray*} The reason why we adopt $\widetilde {B}^*$ instead of $B^*$ will become clear from the following Theorem \ref{bimp}, Remark \ref{3.6} and Proposition \ref{abc}.


The following result is straightforward.
\begin{pro}
If $B$ is a Rota-Baxter operator of weight $\lambda$ on the Lie bialgebra $(\g,\g^*)$. Then $\widetilde {B}^*$ is a Rota-Baxter operator of weight $\lambda$ on the Lie bialgebra $(\g^*,\g)$.
\end{pro}
A Lie bialgebra $(\g,\g^*)$ is naturally a matched pair of Lie algebras. In this case, the representation of $\g$ on $\g^*$ is given by the coadjoint representation  $\ad^*$ of the Lie algebra $\g$ on $\g^*$ and the representation of $\g^*$ on $\g$ is given by the coadjoint representation $\add^*$ of the Lie algebra $\g^*$ on $\g$.

\begin{thm}\label{bimp}
 With above notations,  $(\g,\g^*,B)$ is a Rota-Baxter Lie bialgebra of weight $\lambda$ if and only if $((\g,B),(\g^*,\widetilde {B}^*);\ad^*,\add^*)$ is a matched pair of Rota-Baxter Lie algebras of the same weight.
\end{thm}
\begin{proof}
By Theorem \ref{thm:coadjoint},  $(\g^*,\widetilde {B}^*,\ad^*)$ is a representation of the Rota-Baxter Lie algebra $(\g,B)$, and $(\g,B,\add^*)$ is a representation of the Rota-Baxter Lie algebra $(\g^*,\widetilde {B}^*)$.  Moreover, it is well-known that  $(\g,\g^* )$ is a  Lie bialgebra if and only if $(\g,\g^*;\ad^*,\add^*)$ is a matched pair of  Lie algebras. Therefore, $(\g,\g^*,B)$ is a Rota-Baxter Lie bialgebra of weight $\lambda$ if and only if $((\g,B),(\g^*,\widetilde {B}^*);\ad^*,\add^*)$ is a matched pair of Rota-Baxter Lie algebras of the same weight.
\end{proof}

\begin{rmk}\label{3.6}
Let $(\g,\g^*)$ be a Lie bialgebra, $B:\g\to \g$ and $B^*:\g^*\to \g^*$  Rota-Baxter operators of weight $\lambda$. Then in general $((\g,B),(\g^*,B^*);\ad^*,\add^*)$ is not a matched pair of Rota-Baxter Lie algebras of weight $\lambda$.
In fact, let $C=B^*$ in \eqref{RBmp1} and  take pairing with $y\in \g$. The left hand side of \eqref{RBmp1} equals
\[\langle B^*(\ad_{Bx}^*\xi+\ad_x^*B^*\xi+\lambda\ad_x^*\xi),y\rangle=-\langle \xi,[Bx,By]_\g+B[x,By]_\g+\lambda[x,By]_\g\rangle,\]
and the right hand side of \eqref{RBmp1} amounts to
\[\langle \ad_{Bx}^*B^*\xi,y\rangle=-\langle \xi,B[Bx,y]_\g\rangle.\]
So in general,   the relation \eqref{RBmp1} does not hold, and  $((\g,B),(\g^*,B^*);\ad^*,\add^*)$ is not a matched pair of Rota-Baxter Lie algebras of weight $\lambda$.
\end{rmk}

\begin{rmk} Let $(\g,\g^*,B)$ be  a Rota-Baxter Lie bialgebra of weight $\lambda$. By Theorem \ref{bimp} and  Theorem \ref{demp}, it induces    a descendent matched pair
 $(\g_B,\g^*_{\widetilde {B}^*};\rho_{(B,\widetilde {B}^*)},\mu_{(B,\widetilde {B}^*)})$ of Lie algebras. By \eqref{eq:rep1} and \eqref{eq:rep2}, the representation $\rho_{(B,\widetilde {B}^*)}:\g_B\to\gl(\g^*)$ and $\mu_{(B,\widetilde {B}^*)}:\g^*_{\widetilde {B}^*}\to\gl(\g)$ are given by
$$
\rho_{(B,\widetilde {B}^*)}(x)(\xi)=\ad_{Bx}^*\xi-\ad_x^* B^*\xi,\quad \mu_{(B,\widetilde {B}^*)}(\xi)(x)=\add_{\xi}^*Bx-\add_{B^*\xi}^*x.
$$
It is straightforward to check that $\rho_{(B,\widetilde {B}^*)}$ and $\mu_{(B,\widetilde {B}^*)}$ are not the coadjoint representations of $\g_B$ and $\g^*_{\widetilde {B}^*}$. Therefore, the matched pair
 $(\g_B,\g^*_{\widetilde {B}^*};\rho_{(B,\widetilde {B}^*)},\mu_{(B,\widetilde {B}^*)})$ does not come from a Lie bialgebra.

\end{rmk}

The following result tells us that the Drinfeld double of a Rota-Baxter Lie bialgebra is still a Rota-Baxter Lie bialgebra of the same weight.

\begin{pro}\label{abc}
Let $(\g,\g^*,B)$ be  a Rota-Baxter Lie bialgebra of weight $\lambda$. Then $(\mathfrak{d},\mathfrak{d}^*,\huaB)$ is  a Rota-Baxter Lie bialgebra of weight $\lambda$,
where $(\mathfrak{d},\mathfrak{d}^*)$ is the Lie bialgebra given by Example \ref{doubleex}, and $\huaB: \g\oplus \g^*\to \g\oplus \g^*$ is the linear map defined by
\begin{equation}\label{eq:Bdouble}\huaB( x+\xi) = Bx-\lambda\xi-B^*\xi,\quad x\in \g,~\xi\in \g^*.\end{equation}

\end{pro}

\begin{proof}
First by Theorem \ref{bimp}, $((\g,B),(\g^*,\widetilde {B}^*);\ad^*,\add^*)$ is a matched pair of Rota-Baxter Lie algebras. By Proposition \ref{RBon}, the linear map $\huaB$
is a Rota-Baxter operator of weight $\lambda$ on the Drinfeld double  $\mathfrak{d}:=\g\bowtie \g^*$.

\emptycomment{
Since $B$ and $-\lambda\id-B^*$ are Rota-Baxter operators of weight $\lambda$ on $\g$ and $\g^*$, it suffices to check the following relation
\[[Bx,-\lambda\xi-B^*\xi]_{\g\bowtie \g^*}=\tilde{B}([Bx,\xi]_{\g\bowtie \g^*}+[x,-\lambda\xi-B^*\xi]_{\g\bowtie \g^*}+\lambda[x,\xi]_{\g\bowtie \g^*}).\]
Comparing the  $\g$ and $\g^*$ components, the above relation is equivalent to
\begin{eqnarray}
\label{B1}\lambda\ad_{\xi}^* Bx+\ad_{B^*\xi}^* Bx&=&
B(-\ad_{\xi}^*Bx+\ad_{B^*\xi}^* x);\\
\label{B2}-\lambda\ad_{Bx}^*\xi-\ad_{Bx}^*B^*\xi&=&
(-\lambda\id-B^*)(\ad_{Bx}^*\xi-\ad_x^*B^*\xi).
\end{eqnarray}
Taking pairing with $\eta\in \g^*$, we have
\begin{eqnarray*}
\langle \lambda\ad_{\xi}^* Bx+\ad_{B^*\xi}^* Bx,\eta\rangle &=&-\langle x,\lambda B^*[\xi,\eta]_{\g^*}+B^*[B^*\xi,\eta]_{\g^*}\rangle,\\
\langle B(-\ad_{\xi}^*Bx+\ad_{B^*\xi}^* x),\eta\rangle&=&\langle x, B^*[\xi,B^*\eta]_{\g^*}-[B^*\xi,B^*\eta]_{\g^*}\rangle.
\end{eqnarray*}
So \eqref{B1} holds if and only if $B^*$ is Rota-Baxter operator of weight $\lambda$ on $\g^*$. Similarly, \eqref{B2} holds if and only if $B$ is a Rota-Baxter operator of weight $\lambda$ on $\g$. We proved $\mathrm{(a)}$.
}
Since $\widetilde {B}^*=-\lambda \id-B^*$ is a Rota-Baxter operator of weight $\lambda$ on $\g^*$, it is also a Rota-Baxter operator on the Lie algebra $\overline{\g^*}=(\g^*,-[\cdot,\cdot]_{\g^*})$. Since $\mathfrak{d}^*$ is a direct sum Lie algebra, it is obvious that  $-\lambda\id-\huaB^*=\huaB$ is  a Rota-Baxter operator of weight $\lambda$ on the dual Lie algebra $\mathfrak{d}^*=\g\oplus \overline{\g^*}$. Consequently, $(\mathfrak{d},\mathfrak{d}^*,\huaB)$ is  a Rota-Baxter Lie bialgebra of weight $\lambda$.
\end{proof}

Factorizable Lie bialgebras provide a class of examples of Rota-Baxter Lie bialgebras.

\begin{thm}\label{pro:FL}
Let $(\g,\g^*_r)$ be a factorizable Lie bialgebra with $I=r_+-r_-$. Then $(\g,\g_r^*,B)$ is a Rota-Baxter Lie bialgebra of weight $\lambda$, where $B:=\lambda r_-\circ I^{-1}$ is given by \eqref{eq:defiB}.
\end{thm}

\begin{proof}

It is obvious that $\widetilde {B}^*=-\lambda\id-B^*=\lambda I^{-1}\circ r_-$. Moreover,  by the facts that $\frac{1}{\lambda}I: \g_r^*\to (\g,[\cdot,\cdot]_{B})$ is an isomorphism of Lie algebras (see Corollary \ref{liebi}) and $r_-:\g^*_r\to \g$ is a Lie algebra homomorphism, we have
\begin{eqnarray*}
&&\lambda I^{-1}r_-([\lambda I^{-1} r_-\xi,\eta]_r+[\xi,\lambda I^{-1}r_-\eta]_r+\lambda [\xi,\eta]_r)\\&=&
\lambda  I^{-1}([\lambda r_-I^{-1} r_-\xi,r_-\eta]_\g+[r_-\xi,\lambda r_-I^{-1}r_-\eta]_\g+\lambda[r_-\xi,r_-\eta]_\g)\\
&=&\lambda I^{-1}([r_-\xi,r_-\eta]_{B})\\
&=&[\lambda I^{-1} r_-\xi,\lambda I^{-1}r_-\eta]_r,
\end{eqnarray*}
which implies that $\widetilde {B}^*$
  is a Rota-Baxter operator of weight $\lambda$ on $\g^*_r$. Therefore, $(\g,\g_r^*,B)$ is a Rota-Baxter Lie bialgebra of weight $\lambda$.
\end{proof}

\begin{cor}
Let $(\g,\g^*_r)$ be a factorizable Lie bialgebra with $I=r_+-r_-$. Then have the following commutative diagram of Lie algebra homomorphisms:	
\[
		\vcenter{\xymatrix{
			\cdots\g^*_{\widetilde{B^*}^{k}}\ar[d]^{\frac{1}{\lambda} I}_{\cong} \ar[r]^(0.55){-\lambda \id-B^*}\ar[dr]^{r_-} &\cdots \ar[r]^{} &\g^*_{\widetilde{B^*}} \ar[d]^{\frac{1}{\lambda} I}_{\cong} \ar[r]^(0.55){-\lambda \id-B^*} \ar[dr]^{r_-}&\g_r^* \ar[d]^{\frac{1}{\lambda}I}_{\cong} \ar[r]^(0.55){-\lambda \id-B^*} \ar[dr]^{r_-}&\g_I^* \ar[d]^{\frac{1}{\lambda}I}_{\cong}\\ \cdots \g_{B^{k+1}}\ar[r]^{B}&\cdots \ar[r]^{}&
			\g_{B^2}\ar[r]^{B} &\g_B \ar[r]^{B} &\g,
		}}
	\]
where  $B=\lambda r_-\circ I^{-1}$ and $\widetilde{B^*}=\lambda I^{-1}\circ r_-$  and $\g_{B^k}$ is the descendent Lie algebra of the Rota-Baxter operator $B$ on $\g_{B^{k-1}}$.
\end{cor}
\begin{proof}
 We prove that $\frac{1}{\lambda} I: \g^*_{\widetilde{B^*}^{k}}\to \g_{B^{k+1}}$ is a Lie algebra isomorphism by induction on $k$ $(k\geq 0)$. First observe that $\frac{1}{\lambda} I: \g^*_r\to \g_B$ is a Lie algebra isomorphism, which is the case of $k=0$. Assuming the claim holds for $k-1$,
   then we have
\begin{eqnarray*}
\frac{1}{\lambda}I[\xi,\eta]_{\widetilde{B^*}^{k}}&=&\frac{1}{\lambda}I\big([\lambda I^{-1} r_- \xi,\eta]_{\widetilde{B^*}^{k-1}}+[\xi, \lambda I^{-1} r_-\eta]_{\widetilde{B^*}^{k-1}}+\lambda [\xi,\eta]_{\widetilde{B^*}^{k-1}}\big)\\ &=&[r_-\xi, \frac{1}{\lambda}I \eta]_{B^{k}}+[\frac{1}{\lambda} I \xi, r_-\eta]_{B^{k}}+\lambda[\frac{1}{\lambda} I \xi, \frac{1}{\lambda} I \eta]_{B^{k}}\\ &=&[\frac{1}{\lambda} I \xi, \frac{1}{\lambda} I \eta]_{B^{k+1}},
\end{eqnarray*}
which implies  that it holds for $k$.  All the other facts are obvious.
\end{proof}

\begin{ex}
Let $(\g,\g^*)$ be a Lie bialgebra. Then $(\g,\g^*,B)$ is a Rota-Baxter Lie bialgebra of weight $\lambda$, where  the linear map $B:\g\to \g$ is defined by $B(x)=-\lambda x$. Note that $-\lambda \id-B^*=0$.
\end{ex}

\begin{ex}\label{fre}
Consider the Lie bialgebra  $(\mathfrak{d},\mathfrak{d}^*)$ given in Example \ref{doubleex}, where $\frkd= \g\bowtie \g^*$. Then the linear map
\[B: \g\bowtie \g^*\mapsto \g\bowtie \g^*,\qquad x+\xi\mapsto -\lambda \xi,\]
is a Rota-Baxter operator of weight $\lambda$ on the Lie bialgebra $(\mathfrak{d},\mathfrak{d}^*)$. In fact, it is obvious that $B$ is a Rota-Baxter operator of weight $\lambda$ on the Lie  algebra $ \mathfrak{d}$. On the other hand,  we have
\[(-\lambda\id-B^*)( x+\xi)=-\lambda\xi,\]
which implies that $-\lambda\id-B^*$ is also a Rota-Baxter operator of weight $\lambda$ on the Lie algebra $\frkd^*$. Therefore, $(\mathfrak{d},\mathfrak{d}^*,B)$ is a Rota-Baxter Lie bialgebra   of weight $\lambda$.
\end{ex}

\subsection{Rota-Baxter Lie bialgebras and Manin triples of Rota-Baxter Lie  algebras}

A well-known result regrading Lie bialgebras is that there is a one-one correspondence between Lie bialgebras and Manin triples. A {\bf Manin triple} of Lie algebras  is a triple $((\mathfrak{d},S),\g,\h)$, where $(\mathfrak{d},S)$ is a quadratic Lie algebra, $\g$ and $\h$ are Lie algebras such that
\begin{itemize}
\item [\rm (i)] $\g$ and $\h$ are Lie subalgebras;
\item [\rm (ii)]  $\mathfrak{d}=\g\oplus \h$ as vector spaces;
\item [\rm (iii)] both $\g$ and $\h$ are isotropic with respect to the nondegenerate invariant symmetric bilinear form $S$.
\end{itemize}
Given a Lie bialgebra $(\g,\g^*)$, the triple $((\g\bowtie \g^*,S),\g,\g^*)$   is a Manin triple, where $S$ is given by
\begin{equation}\label{eq:sp}
  S(x+\xi,y+\eta)=\xi(y)+\eta(x).
\end{equation} Conversely, given a Manin triple $((\mathfrak{d},S),\g,\h )$, identifying $\h$ with $\g^*$ by using the nondegenerate invariant symmetric bilinear form $S$, we obtain a Lie bialgebra $(\g,\g^*)$.

Now we    introduce the notion of   Manin triples of Rota-Baxter Lie algebras of weight $\lambda$ using quadratic Rota-Baxter Lie algebras of weight $\lambda$ given in Definition \ref{defi:qua}.

\begin{defi}
A {\bf Manin triple of Rota-Baxter Lie algebras of weight $\lambda$} consists of a triple   $((\huaG,\huaB,S), (\g,B), (\h,C))$, where $(\huaG,\huaB,S)$ is a quadratic Rota-Baxter Lie algebra of weight $\lambda$,  $(\g,B)$ and $(\h,C)$ are Rota-Baxter Lie algebras of weight $\lambda$ such that
\begin{itemize}
\item [\rm (i)]  $(\g,B)$ and $(\h,C)$ are Rota-Baxter Lie subalgebras, i.e. $\g$ and $\h$ are Lie subalgebras of $\huaG$ and $\huaB|_{\g}=B,~\huaB|_{\h}=C$;
\item [\rm (ii)]  $\huaG=\g\oplus \h$ as vector spaces;
\item [\rm (iii)] both $\g$ and $\h$ are isotropic with respect to the nondegenerate symmetric invariant bilinear form $S$.
\end{itemize}
\end{defi}

Similar to the classical case, we have the following result.
\begin{thm}
There is a one-one correspondence between Manin triples of Rota-Baxter Lie algebras of weight $\lambda$ and Rota-Baxter Lie bialgebras
of the same weight.
\end{thm}
\begin{proof}
Let $(\g,\g^*,B)$ be a Rota-Baxter Lie bialgebra. Then $(\g,B)$ and $(\g^*,\widetilde {B}^*)$ are Rota-Baxter Lie algebras of weight $\lambda$. By Proposition \ref{abc}, $(\g\bowtie \g^*,\huaB)$
is a Rota-Baxter Lie algebra of weight $\lambda$, where $\huaB$ is defined by \eqref{eq:Bdouble}. Moreover, it is straightforward to deduce that $(\g\bowtie \g^*,\huaB,S)$ is a quadratic Rota-Baxter Lie algebra of weight $\lambda$, where $S$ is given by \eqref{eq:sp}. Consequently,  $((\g\bowtie \g^*,\huaB,S),(\g,B),(\g^*,-\lambda\id-B^*))$ is a Manin triple of Rota-Baxter Lie algebras of weight $\lambda$.

Conversely, let $((\huaG,\huaB,S), (\g,B), (\h,C))$ be a Manin triple of Rota-Baxter Lie algebras of weight $\lambda$. Then similar to the classical argument, first we can  identify   $\h$ with $\g^*$ by the nondegenerate bilinear form $S$, and  we obtain a Lie bialgebra $(\g,\g^*)$. Then using the invariant condition \eqref{RBmanin}, we identify $C$ with $-\lambda\id-B^*$. Consequently, both $(\g,B)$ and $(\g^*,-\lambda\id-B^*)$ are Rota-Baxter Lie algebras of weight $\lambda$, i.e. $(\g,\g^*,B)$ is a Rota-Baxter Lie bialgebra.
 \end{proof}

 For a Manin triple $((\mathfrak{d},S),\g,\h)$ of Lie algebras, define $\huaB:\mathfrak{d}\to \mathfrak{d}$ by $\huaB(x+\xi)=-\lambda \xi$, where $x \in \g,~\xi\in\h$. Then $\huaB$ is a Rota-Baxter operator of weight $\lambda$ on $\mathfrak{d}$. Furthermore, it is direct to show that $\big((\mathfrak{d},\huaB,S),(\g,0),(\h,-\lambda \id)\big)$ is a Manin triple of Rota-Baxter Lie algebras of weight $\lambda$. Two explicit examples coming from this construction are given as follows.
\begin{ex}
Let $sl(n,\mathbb{C})$ be the Lie algebra of  $n\times n$ traceless complex matrices. Consider its Iwasawa decomposition
\[sl(n,\mathbb{C})=su(n)\oplus sb(n,\mathbb{C}),\]
where $su(n):=\{X\in sl(n,\mathbb{C}); X+\bar{X}^T=0\}$ and $sb(n,\mathbb{C})$ is the Lie algebra of  all $n\times n$ traceless upper triangular complex matrices with real diagonal entries. With the scalar product $S (X,Y)=\mathrm{Im}(tr(XY))$ on the Lie algebra $sl(n,\mathbb{C})$, i.e.  the imaginary part of the trace of $XY$, the triple $((sl(n,\mathbb{C}),S),su(n),sb(n,\mathbb{C}))$ forms a Manin triple (\cite{Lu2}).

The linear map \[\huaB:sl(n,\mathbb{C})\to sl(n,\mathbb{C}),\qquad X=(x_{ij})\mapsto -\lambda A=-\lambda(a_{ij}),\]
where $a_{ii}=\frac{x_{ii}+\overline{x_{ii}}}{2}, a_{ij}=x_{ij}+\overline{x_{ji}}, i< j$ and $a_{ij}=0, i> j$,
is a Rota-Baxter operator of weight $\lambda$ on $sl(n,\mathbb{C})$. In fact, we have
\[\huaB(X)=-\lambda A, \]
if $X\in sl(n,\mathbb{C})$ has the decomposition $X=C+A$ for $C\in su(n)$ and $A\in sb(n,\mathbb{C})$. In fact, $C=(c_{ij})$ is defined by $c_{ii}=\frac{x_{ii}-\overline{x_{ii}}}{2}, c_{ij}=-\overline{x_{ji}}, i<j$ and $c_{ij}=x_{ij}, i>j$. It is obvious $\huaB|_{su(n)}=0$ and $\huaB|_{ sb(n,\mathbb{C})}=-\lambda \id$ are Rota-Baxter operators of weight $\lambda$. Therefore, $((sl(n,\mathbb{C}),\huaB,S),(su(n),0),(sb(n,\mathbb{C}),-\lambda \id))$ is a Manin triple of Rota-Baxter Lie algebras of weight $\lambda$.
This example can be generalized to the Iwasawa decomposition of   an arbitrary semi-simple Lie algebra.
\end{ex}
\begin{ex}
Let $\g=sl(n,\mathbb{C})$ and consider   the direct sum Lie algebra $\g\oplus \g=sl(n,\mathbb{C})\oplus sl(n,\mathbb{C})$.
Define \begin{eqnarray*}\g_{diag}&=&\{(X,X);X\in sl(n,\mathbb{C})\};\\ \g^*_{st}&=&\{(Y+X_+,-Y+X_-);Y\in \h,X_+\in n_+,X_-\in n_-\},\end{eqnarray*}
where $\h,n_+,n_-$ are the Lie algebras of diagonal, strictly upper and strictly lower triangular matrices in $sl(n,\mathbb{C})$. With respect to the scalar product
\[\langle(X_1,Y_1),(X_2,Y_2)\rangle_{\g\oplus \g}=\mathrm{Im}(tr(X_1X_2))-\mathrm{Im}(tr(Y_1Y_2)),\]
the triple $(\g\oplus \g,\g_{diag},\g^*_{st})$ consists a Manin triple. Moreover, the map
\begin{eqnarray*}
\huaB: sl(n,\mathbb{C})\oplus sl(n,\mathbb{C})&\to& sl(n,\mathbb{C})\oplus sl(n,\mathbb{C}),\\
(X,Y)&\mapsto& -\lambda\big(\frac{1}{2}(X-Y)_0+(X-Y)_+,-\frac{1}{2}(X-Y)_0-(X-Y)_-\big),
\end{eqnarray*}
is a Rota-Baxter operator of weight $\lambda$ on this Manin triple, where $(X-Y)_0,(X-Y)_+$ and $(X-Y)_-$ are the diagonal part, strictly upper triangular part and strictly lower triangular part of the matrix $X-Y$, respectively. In fact, taking into consideration of the decomposition $\g\oplus \g=\g_{diag}\oplus \g_{st}^*$, the linear map $\huaB$ is actually the projection to the $\g_{st}^*$- component multiplied by $-\lambda$.
\end{ex}

\emptycomment{

\yh{Maybe in this case it is easier to study coboundary case?}

\hl{In terms of $\Delta: \g\to \wedge^2 \g$, $B$ is a Rota-Baxter operator on the Lie bialgebra $(\g,\Delta)$  if and only if $B:\g\to \g$ is a Rota-Baxter operator of weight $\lambda$ and
\[(B\otimes B)\circ \Delta(x)=(B\otimes \id+\id\otimes B+\lambda \id\otimes \id)\circ \Delta(B(x)),\qquad \forall x\in \g.\]
If in particular, $\Delta(x)=[x,r]_\g$ for $r\in \wedge^2 \g$ such that $[r,r]_\g=0$, then
we have
\[(B\otimes B)([x,r]_\g)=(B\otimes \id+\id\otimes B+\lambda \id\otimes \id)\circ [Bx,r]_\g,\qquad \forall x\in \g.\]
It does not look simpler for the triangular case when $\Delta(x)=[x,r]$.}

\textcolor{red}{Problem 1: this also looks mysterious.}

\comment{If you view $B$ as Nijenhuis (very similar), is it possible using Poisson-Nijenhuis idea to analysis this problem? i.e. for example, is possible $Br^\sharp$ is a relative RB of weight 0? }

}

\section{Rota-Baxter Lie groups and factorizable Poisson Lie groups}\label{sec:rbg}
In this section, we show that factorizable Poisson Lie groups are in one-to-one correspondence with quadratic Rota-Baxter Lie groups. Moreover, a Rota-Baxter Lie group gives rise to a matched pair of Lie groups.

\subsection{Factorizable Poisson Lie groups and quadratic Rota-Baxter Lie groups}
The notion of Rota-Baxter Lie groups was introduced in \cite{GLS}, whose differentiations are Rota-Baxter Lie algebras of weight 1.

\begin{defi}{\rm(\cite{GLS})}
 A {\bf Rota-Baxter operator} on a Lie group $G$ is a smooth map $\B:G\longrightarrow G$ such that
\begin{eqnarray}\mlabel{RB}
\B(g_1)\B(g_2)=\mathfrak{B}(g_1\Ad_{\B(g_1)} g_2),\quad \forall g_1,g_2\in G.
\end{eqnarray}
The pair $(G,\B)$ is called a {\bf Rota-Baxter Lie group}.
\end{defi}

Let $(G,\B)$ be a Rota-Baxter Lie group. Then there is a new group structure $\star$ on $G$ given by
$$
g\star h=g\Ad_{\B(g)} h.
$$
The new Lie group $(G,\star)$ is called  the descendent Lie group, and  denoted by $G_{\B}$.
Moreover, $\B$ is a Lie group homomorphism:
\begin{equation}\label{eq:Bhomo}
  \B(g\star h)=\B(g)\B(h).
\end{equation}

\begin{thm}{\rm(\cite{GLS})}
  Let $(G,\B)$ be a Rota-Baxter Lie group, and $\g$ the Lie algebra of the Lie group $G$. Then $(\g,B:=\B_{*e})$ is a Rota-Baxter Lie algebra of weight $1$, where $e$ is the identity of $G.$
\end{thm}
   Recall from Theorem \ref{converse} that given a quadratic Rota-Baxter Lie algebra $(\g,B,S)$, one has a linear isomorphism $\huaI_S: \g^*\to \g$ defined by $\langle \huaI_S^{-1}(x),y\rangle=S(x,y)$. Using $\huaI_S$, we pull-back the Lie algebra structure $[\cdot,\cdot]_B$ on $\g$ to define a Lie algebra structure on $\g^*$. Denote it by $\g^*_{S}$. Then $\huaI_S:\g^*_S\to \g_B$ becomes a Lie algebra isomorphism. Let $G^*$ be the simply connected Lie group whose Lie algebra is $\g^*_S$. Now we introduce the notion of quadratic  Rota-Baxter Lie groups, which will be used to characterize factorizable Poisson Lie groups.
\begin{defi}\label{qrblg}
A {\bf quadratic Rota-Baxter Lie group} is a triple   $(G,\B,S)$, where  $(G,\B)$ is a Rota-Baxter Lie group,   $S\in \otimes^2 \g^*$ such that $(\g,B:=\B_{*e},S)$ is a quadratic Rota-Baxter Lie algebra and the linear isomorphism $\huaI_S: \g^*_S\to \g$ can be lifted to a global diffeomorphism $J_S: G^*\to G$.
\end{defi}

A Poisson Lie group $G$ is called {\bf quasitriangular} (\cite{WX}) if its corresponding Lie bialgebra $(\g,\g^*)$ is quasitriangular. Let $R_+,R_-: G^*\to G$ be the lifted   Lie group homomorphisms of the Lie algebra homomorphisms $r_+,r_-:\g^*\to \g$, where $G^*$ is the simply connected dual Poisson Lie group of $G$. For simplicity, we always suppose $G$ is connected and simply connected.

Since $r_+^*=-r_-$ is a Lie algebra anti-homomorphism,  the map $R_+$ is an anti-Poisson map. In summary, both $R_+$ and $R_-$ are Lie group homomorphisms and anti-Poisson maps.

Define
\begin{eqnarray}\label{J}
J: G^*\to G,\quad u\mapsto R_+(u)R_-(u)^{-1},\quad u\in G^*.
\end{eqnarray}
A quasitriangular Poisson Lie group $G$ is called {\bf factorizable} if $J$ is a global diffeomorphism. Denote a factorizable Poisson Lie group  by $(G, J)$.

\begin{thm}\label{GFL}
Let $(G,J)$ be a factorizable Poisson Lie group, $(\g,\g^*_r)$ the corresponding factorizable Lie bialgebra with $I=r_+-r_-$. Then $(G,\B:=R_-\circ J^{-1},S_I)$ is a quadratic Rota-Baxter Lie group, where $S_I\in \otimes^2 \g^*$ is defined by $S_I(x,y):=\langle I^{-1}x, y\rangle$ for $x,y\in \g$.

   Moreover, $J: G^*\to G_\B$ is a Lie group isomorphism.

\end{thm}
\begin{proof}
 Since $R_+,R_-:G^*\to G$ are Lie group homomorphisms, we have
\begin{eqnarray}\label{Jhomo}
 J(J^{-1}g J^{-1} h)&=&R_+(J^{-1} g J^{-1} h)R_-(J^{-1} g J^{-1} h)^{-1} \\ &=&\nonumber
R_+(J^{-1} g)R_+(J^{-1}h) R_-(J^{-1} h)^{-1} R_-(J^{-1} g)^{-1}
\\ &=& \nonumber J(J^{-1} g)  R_-(J^{-1} g) J(J^{-1} h)  R_-(J^{-1} h) R_-(J^{-1} h)^{-1} R_-(J^{-1} g)^{-1}
\\ \nonumber&=& g R_-(J^{-1} g) h R_-(J^{-1} g)^{-1}.
\end{eqnarray}
Then we have
  \begin{eqnarray*}
   \B(g\B(g)h\B(g)^{-1})&=&R_-J^{-1}(gR_-(J^{-1} g)h R_-(J^{-1}g)^{-1})\\
   &=&R_-(J^{-1}gJ^{-1} h)\\
   &=&R_-(J^{-1}g)R_-(J^{-1}h)\\
   &=&\B(g)\B(h).
  \end{eqnarray*}
 Therefore, $\B$ is Rota-Baxter operator  on $G$. The fact that $(G,\B,S_I)$ is a quadratic Rota-Baxter Lie group (Definition \ref{qrblg})    follows from Theorem \ref{FL}.

Moreover, for any $u,v\in G^*$, setting $g=Ju$ and $h=Jv$, by \eqref{Jhomo} we get
\[J(uv)=Ju \B(Ju) Jv \B(Ju)^{-1}=(Ju)\star (Jv).\]
Since $J$ is a global diffeomorphism, it is a Lie group isomorphism from the Lie group $G^*$ to the descendent Lie group $G_\B$.
\end{proof}

By  \cite[Proposition 2.4]{GLS}, if $\B:G\to G$ is a Rota-Baxter operator on the Lie group $G$, then $\widetilde{\B}:G\to G$ defined by $\widetilde{\B}(g)=g^{-1}\B(g^{-1})$ is also a Rota-Baxter operator on $G$.
\begin{cor}
Let $(G,J)$ be a factorizable Poisson Lie group. Then
 the map $\widetilde{\B}:G\to G$  given by
 $$ \widetilde{\B}(g) = R_+(J^{-1}(g^{-1}))$$ is a Rota-Baxter operator  on $G$.
\end{cor}

\begin{proof}
 We have
\begin{eqnarray*}
\widetilde{\B}(g)&=&g^{-1}R_-(J^{-1}(g^{-1}))\\
&=&J(J^{-1}(g^{-1}))R_-(J^{-1}(g^{-1}))\\
&=&R_+(J^{-1}(g^{-1}))R_-(J^{-1}(g^{-1}))^{-1}R_-(J^{-1}(g^{-1}))\\
&=&R_+(J^{-1}(g^{-1})).
\end{eqnarray*}
By Theorem \ref{GFL}, $\B$ is a Rota-Baxter operator. Therefore $\widetilde{\B}:G\to G$ is also a Rota-Baxter operator.
\end{proof}
\begin{cor}
Let $(G,J)$ be a factorizable Poisson Lie group, and $\B$ the induced Rota-Baxter operator on $G$. Then the descendent Lie group $G_{\B}$   is a Poisson Lie group with the Poisson structure $\pi_{G_{\B}}=J_*(\pi_{G^*})$. Its infinitesimal Lie bialgebra is $\big(\g_B,(\g^*,[\cdot,\cdot]_{I})\big)$ given in Corollary \ref{liebi}.

\end{cor}



At the end of this section, we show that a quadratic Rota-Baxter Lie group gives rise to a factorizable Poisson Lie group.

\begin{thm}\label{thm:rbg-fac}
Let $(G,\B,S)$ be a quadratic Rota-Baxter Lie group and $J_S:G^*\to G$ the induced diffeomorphism.
Then $G$ is a factorizable Poisson Lie group with
\[R_-:=\B\circ J: G^*\to G,\qquad R_+:=\B_+\circ J:G^*\to G,\]
where $\B_+(g)=g\B(g)$.
\end{thm}

\begin{proof}
By definition, the infinitesimal $(\g,B:=\B_{*e},S)$ is a quadratic Rota-Baxter  Lie algebra of weight 1. By Theorem \ref{converse}, we obtain that $(\g,\g^*_r)$ is a factorizable Lie bialgebra. Note that
\begin{eqnarray*}
      (R_+)_{*e}&=&(B+\id)\circ \huaI_S=r_+,\\
       (R_-)_{*e}&=&B\circ \huaI_S=r_-,\\
      (J_S)_{*e}&=&\huaI_S.
   \end{eqnarray*}
    So integrating this factorizable Lie bialgebra (\cite{Lu2}), we get a factorizable Poisson Lie group structure on $G$.
\end{proof}

As the Lie group level of Example \ref{doubleex}, we have
\begin{ex}
With notations in Example \ref{doubleex}, let $G$ and $G^*$ be the simply connected Lie groups of $\g$ and $\g^*$. Suppose that $G$ is complete. Then the Lie group integrating the Lie algebra  $\mathfrak{d}$ is $D=G\bowtie G^*$. Denote by $D^*$ the simply connected dual Poisson Lie group. In fact, $D^*=G\times \overline{G^*}$, the direct Lie group of $G$ and $\overline{G^*}$, where $\overline{G^*}$ is the manifold $G^*$ with the group structure being $s\cdot_{\overline{G^*}} t=ts$ for $s,t\in G^*$.
Note that $r_+,r_-$ lift naturally to $R_+,R_-: D^*\to D$, which is given by
\[R_+(g,u)=g,\qquad R_-(g,u)=u^{-1},\qquad \forall g\in G,u\in G^*.\]
So $D$ is a factorizable Poisson Lie group.  The induced Rota-Baxter operator $\B $ on the Lie group $D$ is given by $$\B(g,u)=u^{-1}.$$
\end{ex}

\subsection{Rota-Baxter Lie groups and matched pairs of Lie groups}
In this subsection, we obtain a matched pair of Lie groups from a Rota-Baxter Lie group.  Matched pairs of Lie groups, also called double Lie groups, were explored in \cite{Lu2, Majid}.  For matched pair of groups, we refer the readers to \cite{18}. 

A pair of Lie groups $(P,Q)$ is called a {\bf matched pair of Lie groups} (\cite{Majid}) if there is a left action of $P$ on $Q$ and a right action of $Q$ on $P$:
\[P\times Q\to Q,\qquad (p,q)\mapsto  p\triangleright q;\qquad P\times Q\to P\qquad (p,q)\mapsto p\triangleleft q,\]
such that
\begin{eqnarray}
\label{mpg1} p\triangleright (q_1q_2)&=&(p\triangleright q_1)\big((p\triangleleft q_1)\triangleright q_2\big);\\
\label{mpg2} (p_1p_2)\triangleleft q&=&\big(p_1\triangleleft (p_2\triangleright q)\big)   (p_2\triangleleft q).
\end{eqnarray}
Denote a Lie group matched pair by $(P,Q;\triangleright,\triangleleft)$.

The following equivalent characterization of matched pairs of Lie groups is well-known.

\begin{pro}\label{pro:mpg}
  Let $(P,Q;\triangleright,\triangleleft)$ be a matched pair of Lie groups. Then there is a Lie group structure $\cdot_{\bowtie}$ on $Q\times P$ defined by
\[(q_1,p_1)\cdot_{\bowtie}(q_2,p_2)=\big(q_1(p_1\triangleright q_2),(p_1\triangleleft q_2) p_2\big).\]

Conversely, if  $(Q\times P,\cdot)$ is a Lie group such that $Q$ and $P$ are Lie subgroups, then $(P,Q;\triangleright,\triangleleft)$ is a matched pair of Lie groups, where the  left action $\triangleright: P\times Q\to Q$ of $P$ on $Q$ and the right action $\triangleleft: P\times Q\to P$ of $Q$ on $P$ are determined by
\[(1,p)\cdot(q,1)=(p\triangleright q,p\triangleleft q).\]
\end{pro}
 Denote the Lie group $(Q\times P,\cdot_{\bowtie})$ by $Q\bowtie P$.

Let $G$ be a Lie group. Denote by $(G\times G,\cdot_{G\times G})$  the direct product Lie group
with the group structure  given by
\[(g_1,g_2)\cdot_{G\times G} (h_1,h_2)=(g_1h_1,g_2h_2).\]

\begin{pro}\label{doublegroup}
Let $(G,\B)$ be a Rota-Baxter Lie group. Then there is a Lie group structure $\cdot_D$ on   $G\times G$ determined by
\begin{eqnarray*}
(1,g)\cdot_D (1,h)&=&(1, gh);\\
(s,1)\cdot_D (t,1)&=&(s\star t,1)=(s\B(s)t\B(s)^{-1},1);\\
(s,1)\cdot_D (1,g)&=&(s,g);\\
(1,g)\cdot_D (s,1)&=&(gsg^{-1}, \B(gsg^{-1})^{-1}g\B(s)),
\end{eqnarray*}
where $s,t$ are elements in the first component, and $g,h$ are elements in the second component.  Moreover, there is a Lie group isomorphism $\Phi$ from $(G\times G,\cdot_D)$ to the direct product Lie group $(G\times G,\cdot_{G\times G})$  defined by
\begin{eqnarray}\label{Phi}
\Phi (s,g)=(s \B(s)g,\B(s)g).
\end{eqnarray}
\end{pro}
\begin{proof}It is simple to see that $\Phi$ is a diffeomorphism whose inverse is
 given by
\[\Phi^{-1} (g,h)= (gh^{-1},\B(gh^{-1})^{-1}h).\]
Then we claim that the product $\cdot_D$ on $G\times G$ is the pull-back of the direct product Lie group structure on $G\times G$ by $\Phi$, i.e.,
\[e_1\cdot_D e_2=\Phi^{-1}\big(\Phi(e_1)\cdot_{G\times G} \Phi(e_2)\big),\quad \forall e_1,e_2\in G\times G.\]
  First, it is obvious that $(1,g)\cdot_D (1,h)=\Phi^{-1}\big(\Phi(1,g)\cdot_{G\times G} \Phi(1,h)\big)$. By \eqref{eq:Bhomo}, we have
\begin{eqnarray*}
\Phi^{-1}\big(\Phi(s,1)\cdot_{G\times G}\Phi(t,1)\big)&=&\Phi^{-1}\big((s\B(s), \B(s))\cdot_{G\times G}(t\B(t), \B(t))\big)\\ &=&\Phi^{-1}
(s\B(s)t\B(s)^{-1}\B(s)\B(t),\B(s)\B(t))\\ &=&\Phi^{-1}(s\star t \B(s\star t), \B(s\star t))\\ &=&(s,1)\cdot_D (t,1).
\end{eqnarray*}
Moreover, by definition, we have
\begin{eqnarray*}
\Phi^{-1}\big(\Phi(s,1)\cdot_{G\times G}\Phi(1,g)\big)&=&\Phi^{-1}\big((s\B(s),\B(s))\cdot_{G\times G}(g,g)\big)\\ &=&\Phi^{-1}(s\B(s)g,\B(s)g)=(s,g)\\
&=&(s,1)\cdot_D (1,g),
\end{eqnarray*}
and
\begin{eqnarray*}
\Phi^{-1}\big(\Phi(1,g)\cdot_{G\times G} \Phi(s,1)\big)=\Phi^{-1}(gs\B(s),g\B(s))=(gsg^{-1},\B(gsg^{-1})^{-1}g\B(s))=
(1,g)\cdot_D (s,1).
\end{eqnarray*}
Therefore, the Lie group structure $\cdot_D$ on $G\times G$ is the pull-back of the direct product Lie group structure on $G\times G$ by $\Phi$, and $\Phi$ is a Lie group isomorphism.
\end{proof}

By using the associativity, the group structure $\cdot_D$   is given by
\begin{eqnarray}\label{dg}
(s,g)\cdot_D(t,h)&=& (s,1)\cdot_D (1,g)\cdot_D (t,1)\cdot_D (1,h)\\ &=&\nonumber
(s,1)\cdot_D(gtg^{-1}, \B(gtg^{-1})^{-1}g \B(t))\cdot_D(1,h)\\ &=&\nonumber (s,1)\cdot_D(gtg^{-1},1)\cdot_D(1,\B(gtg^{-1})^{-1}g \B(t))\cdot_D(1,h)\\ \nonumber&=&(sgtg^{-1}, \B(gtg^{-1})^{-1}g\B(t)h).
\end{eqnarray}

\begin{thm}\label{thm:mp}
Let $(G,\B)$ be a Rota-Baxter Lie group. Then we have a matched pair of Lie groups $(G,G_{\B};\triangleright, \triangleleft)$,
where the actions are given by
\[g\triangleright s=gsg^{-1},\quad g\triangleleft s=\B(gsg^{-1})^{-1}g\B(s),\quad \forall g\in G,s\in G_\B.\]
Moreover, the corresponding Lie  group   $G_\B\bowtie G$ is exactly the Lie group
$(G\times G, \cdot_D)$ given in Proposition \ref{doublegroup}.
\end{thm}

\begin{proof}

By \eqref{dg}, we get
$$
(s,g)\cdot_D(t,h)=(s(g\triangleright t), (g\triangleleft t) h).
$$
By Proposition \ref{pro:mpg}, we deduce that $(G,G_{\B};\triangleright, \triangleleft)$ is a matched pair of Lie groups.
\end{proof}

Let $(G,J)$ be a factorizable Poisson Lie group. Then $(G,\B=R_-\circ J^{-1})$ is a Rota-Baxter Lie group. By Theorem \ref{thm:mp}, we have a Lie group $G_{\B}\bowtie G$. On the other hand, there is a Lie group $G^*\bowtie G$ on the double of a Poisson Lie group.
\begin{pro}
With the notations above, we have a commutative diagram of Lie group isomorphisms:
\begin{eqnarray}\label{diagram2}
		\vcenter{\xymatrix{
			G^*\bowtie G \ar[d]_(0.45){J\times \id}\ar[r]^(0.45){\Psi} &G\times G\\
			G_{\B}\bowtie G,\ar[ur]_{\Phi}&
		}}
\end{eqnarray}
where $\Phi$ is defined by \eqref{Phi} and $\Psi$ is given by
\[\Psi(s,g)=(R_+(s) g,R_-(s) g),\qquad \forall s\in G^*,g\in G.\]
Moreover, the differential of this diagram is the diagram \eqref{diagram} in the case of $\lambda=1$.
\end{pro}
\begin{proof}
It is known that $\Phi$ is a Lie group homomorphism that integrates $\phi$. By  Theorem \ref{GFL}, $J:G^*\to G_\B$ is the unique Lie group isomorphism integrating $I: \g^*_r\to \g_B$, as $G^*$ is simply connected. So it is direct to see $J\times \id:G^*\bowtie G\to G_\B\bowtie G$ is a Lie group isomorphism whose infinitesimal is $I\oplus \id: \g^*_r\bowtie \g\to \g_B\bowtie \g$ in \eqref{diagram}.
Moreover, by the calculation
\[\Phi\circ (J\times \id)(s,g)=\Phi(Js,g)=(Js R_-J^{-1}(Js)g,R_-J^{-1}(Js)g)=(R_+(s)g,R_-(s)g)=\Psi(s,g),\]
we see that $\Phi\circ (J\times \id)=\Psi$, which is also a Lie group isomorphism such that $\Psi_{*e}=\psi$. Thus we conclude that diagram \eqref{diagram2} is the integration of the diagram \eqref{diagram} in the case of $\lambda=1$.
\end{proof}

Since $\Phi$ is a Lie group isomorphism, it follows that $\im(\Phi|_{G_{\B}})\subset G\times G$ is a Lie subgroup, which is isomorphic to $G_{\B}$. It gives an alternative approach to the factorization of Rota-Baxter Lie groups given in \cite{GLS}. 

\begin{cor}
Let $(G,\B)$ be a Rota-Baxter Lie group. Then for any $g\in G$, there exists a unique decomposition
$g=g_+g_-^{-1}$ with
 $(g_+,g_-)\in \im(\Phi|_{G_{\B}})\subset G\times G$.
\end{cor}
\begin{proof}
We have $g=g\B(g) \B(g)^{-1}$, where $(g\B(g),\B(g))=\Phi(g,1)$ for $g\in G_{\B}$ and this decomposition is unique since $\Phi$ is a diffeomorphism.
\end{proof}

At the end of this subsection, we establish the relation between the Lie group $G_\B\bowtie G$ and the Lie algebra $\g_B\bowtie \g$, where $\g$ is the Lie algebra of $G$ and $B=\B_{*e}$ is the differentiated Rota-Baxter operator of weight 1 on $\g.$
\begin{pro}
Let $(G,\B)$ be a Rota-Baxter Lie group whose infinitesimal Rota-Baxter Lie algebra is $(\g,B=\B_{*e})$. Then the Lie algebra of the Lie group $G_\B\bowtie G$ given in Theorem \ref{thm:mp} is the Lie algebra $\g_B\bowtie \g$ given in Theorem \ref{thm:mpB} and $\Phi_{*e}=\phi$.
\end{pro}
\begin{proof}
Denote by $\Courant{\cdot,\cdot}$ the induced Lie bracket on $\g\oplus \g$  of the Lie group $G_\B\bowtie G$. It is simple to see  that
\[\Courant{(0,x),(0,y)}=(0,[x,y]_\g),\qquad \Courant{(\xi,0),(\eta,0)}=([\xi,\eta]_B,0).\]
Then, by \eqref{dg}, we have
\begin{eqnarray*}
&&\Courant{(0,x),(\xi,0)}\\&=&\frac{d}{d\epsilon}\frac{d}{d\epsilon'}|_{\epsilon,\epsilon'=0}(1, \exp^{\epsilon x})\cdot_D (\exp^{\epsilon' \xi},1)\cdot_D(1,\exp^{-\epsilon x})\\ &=&\frac{d}{d\epsilon}\frac{d}{d\epsilon'}|_{\epsilon,\epsilon'=0}(1,\exp^{\epsilon x})\cdot_D(\exp^{\epsilon' \xi},\exp^{-\epsilon x})
\\ &=&\frac{d}{d\epsilon}\frac{d}{d\epsilon'}|_{\epsilon,\epsilon'=0}
(\exp^{\epsilon x} \exp^{\epsilon' \xi}\exp^{-\epsilon x}, \B(\exp^{\epsilon x} \exp^{\epsilon' \xi}\exp^{-\epsilon x})^{-1}\exp^{\epsilon x}\B(\exp^{\epsilon' \xi})\exp^{-\epsilon x})\\ &=&([x,\xi]_\g, \frac{d}{d\epsilon}\frac{d}{d\epsilon'}|_{\epsilon,\epsilon'=0}\B(\exp^{\epsilon x} \exp^{\epsilon' \xi}\exp^{-\epsilon x})^{-1}+\frac{d}{d\epsilon}\frac{d}{d\epsilon'}|_{\epsilon,\epsilon'=0}\exp^{\epsilon x}\B(\exp^{\epsilon' \xi})\exp^{-\epsilon x})\\ &=&([x,\xi]_\g, -B[x,\xi]_\g+[x,B\xi]_\g)\\
&=&[(0,x),(\xi,0)]_{\g_B\bowtie \g}.
 \end{eqnarray*}
 So we proved $\Courant{\cdot,\cdot}=[\cdot,\cdot]_{\g_B\bowtie \g}$. Namely, the Lie algebra of $G_{\B}\bowtie G$ is $\g_B\bowtie \g$. It is obvious that $\Phi_{*e}=\phi$.
\end{proof}

\end{document}